\documentclass[accepted]{uai2021} 

\usepackage[american]{babel}

\usepackage{natbib} 
\usepackage{mathtools} 
\usepackage{booktabs} 
\usepackage{tikz} 

\usepackage{style}

\title{Learning in Multi-Player Stochastic Games}

\author[1]{\href{mailto:William Brown <w.brown@columbia.edu>?Subject=No-Regret Learning in Multi-Player Stochastic Games}{William Brown}{}} 

\affil[1]{%
    Computer Science Dept.\\
    Columbia University\\
    New York, New York, USA
}

\begin{document}
\maketitle
\begin{abstract}
We consider the problem of simultaneous learning in stochastic games with many players in the finite-horizon setting. While the typical target solution for a stochastic game is a Nash equilibrium, this is intractable with many players. We instead focus on variants of {\it correlated equilibria}, such as those studied for extensive-form games. We begin with a hardness result for the adversarial MDP problem: even for a horizon of 3, obtaining sublinear regret against the best non-stationary policy is \textsf{NP}-hard when both rewards and transitions are adversarial. This implies that convergence to even the weakest natural solution concept---normal-form coarse correlated equilbrium---is not possible via black-box reduction to a no-regret algorithm even in stochastic games with constant horizon (unless $\textsf{NP}\subseteq\textsf{BPP}$). Instead, we turn to a different target: algorithms which {\it generate} an equilibrium when they are used by all players. Our main result is algorithm which generates an {\it extensive-form} correlated equilibrium, whose runtime is exponential in the horizon but polynomial in all other parameters. We give a similar algorithm which is polynomial in all parameters for ``fast-mixing'' stochastic games. We also show a method for efficiently reaching normal-form coarse correlated equilibria in ``single-controller'' stochastic games which follows the traditional no-regret approach. When shared randomness is available, the two generative algorithms can be extended to give simultaneous regret bounds and converge in the traditional sense.
\end{abstract}

\section{Introduction}\label{sec:intro}

Many multi-agent systems, such as financial markets, transportation networks, and video games, involve agents competing in environments where their actions affect their immediate rewards as well as transitions between states in the environment. When opponent strategies are fixed, this resembles a reinforcement learning problem for a single agent.
Stochastic games, also known as Markov games are a popular model for multi-agent reinforcement learning problems (\cite{Littman1994MarkovGA}), and have also been studied extensively throughout economics and computer science (\cite{Solan13743,shoham_leyton-brown_2008}).
They generalize Markov decision processes (MDPs) to many players, where each state is now a game where both the instantaneous rewards and transitions depend on the actions of all players. As is the case throughout game theory, a fundamental question from the perspective of algorithm design is whether some kind of equilibrium can be found efficiently.

The traditional solution concept for a game, often interpreted as a model of rational behavior, is the Nash equilibrium (\cite{Nash48}).  In two-player zero-sum and other restricted classes of normal-form games, Nash equibria can be found efficiently; however, finding one is \textsf{PPAD}-complete for arbitrary games even with only two players (\cite{DGP06,CDT07}), and thus likely computationally intractable.
A more appropriate target in this case is a {\it correlated} equilibrium, introduced by \cite{AUMANN197467}, which is a generalization of a Nash equilibrium where strategies can be correlated across players, and can be efficiently computed in general games (e.g.\ \cite{nisan_roughgarden_tardos_vazirani_2007}).
Correlated equilibria can also be reached through repeated play by agents who use appropriate learning algorithms.
The existence of no-swap-regret dynamics which efficiently converge to correlated equilibria is a celebrated result in the theory of learning in games (\cite{FOSTER199740,hartmascolell,blummansour}). A notable benefit of this approach is that it does not depend on the description length of the game, so long as rewards are computable from an action profile, and thus can be used in many-player games where writing an explicit game description is prohibitive.

The normal-form game model is often insufficient to capture problems of practical interest. The aforementioned results cannot be applied directly to stochastic games, as the strategy space is exponential in the relevant parameters. Yet, as real-world problems often have many players and possibly arbitrary reward structures, it is natural to target correlated equilibria as a solution concept for stochastic games as well.
The starting point for our work is asking whether an efficient convergence result of the same form as \cite{hartmascolell} can be obtained for repeated play of a stochastic game in the finite-horizon setting.

A related setting where similar questions have been studied is that of extensive-form games (EFGs). Several refinements of correlated equilibria have been proposed for EFGs, which differ in when action recommendations are revealed to each agent (\cite{stengel,efcepoly,farina2019coarse}).
Two variants which we will consider are normal-form coarse correlated equilibria (NFCCE) and extensive-form correlated equilibria (EFCE), with the latter contained in the former, which we adapt to finite-horizon stochastic games.
Recent work has led to the development of an algorithm which converges to an EFCE by minimizing an appropriate notion of regret for each agent (\cite{celli2020noregret}). We show that such a black-box reduction cannot work for stochastic games of even constant horizon, as the corresponding online learning problem is hard, and instead design algorithms which converge to correlated equilibria (in a somewhat delicate sense) by directly leveraging information about opponents' strategies.

\subsection{Results and Techniques} 

We assume that players in a finite-horizon stochastic game play for many repeated horizons, or {\it trajectories}, and that rewards and transition dynamics are computed by an oracle when players submit actions simultaneously at a given state. Players receive only {\it bandit} feedback, i.e.\ they do not know what see what rewards or transitions would have occurred if they had selected a different action. 
Longer horizons allow for greater consideration of ``deferred rewards'' for actions, such as in a board game where an early move can become consequential in the endgame; a horizon of one corresponds to a repeated one-shot game. 
Each form of correlated equilibrium we consider is a joint distribution over recommended {\it policies}, which tell each player an action to play at each state. We consider policies which are non-stationary, i.e. they can depend on the time-step.
In a NFCCE, no player can improve their reward by committing to a fixed policy before the trajectory begins or recommendations are revealed. In an EFCE, players receive individual action recommendations only upon reaching a state, and they cannot improve rewards by ``swapping'' their actions based on their recommendations.

Our first result is negative: we show that obtaining sublinear regret against the best non-stationary policy for adversarial MDPs with a horizon of 3 is \textsf{NP}-hard, strengthening previous hardness results (\cite{ExpertsMDP04,NIPS2013_4f284803}) which require the stronger ``LPN hardness'' assumption and hold only when the horizon is approximately the size of the MDP. The adversarial MDP problem is the natural online learning variant of our setting, as each set of opponent policies defines an MDP for a given player, albeit with different rewards and transitions. Assuming $\textsf{NP} \not\subseteq \textsf{BPP}$, this implies that any algorithm which quickly converges to even a NFCCE in a stochastic game with constant horizon cannot be no-regret against arbitrary opponents, ruling out a black-box reduction to reaching a correlated equilibrium as in \cite{hartmascolell} or \cite{celli2020noregret}.

We then turn our attention to designing algorithms which make use of information about the behavior of opponents, namely that they are using the same algorithm. While regret minimization and learning equilibria are often viewed as intimately connected, lower bounds for regret do not necessarily imply barriers for equilibria when opponents are not behaving arbitrarily; in particular, knowledge of ``self-play'' has been used to obtain rates of convergence to correlated equilibria in normal-form games which overcome lower bounds for regret minimization against an arbitrary adversary (\cite{SyrgkanisFast,chen2020hedging}).

Our main result is a decentralized learning algorithm which reaches an EFCE when used by all players, and in particular one where the distribution of recommended action profiles at each state is a product distribution across states.
We observe that {\it computing} an EFCE of this form is straightforward in a centralized model, as it reduces to the problem of finding a correlated equilibrium for a set of normal-form games, each of which can be computed with linear programming or no-swap-regret learning. States at the final time-step are essentially equivalent to normal-form games, and each player will have a {\it value} associated with a given correlated equilibrium representing their average reward at that state-time pair. 
These values can be folded back into rewards at previous time-steps, enabling an inductive computation.
Our main algorithm, PLL, aims to simulate this approach by conducting repeated {\it parallel local learning} at each state. After a number of trajectories which is exponential in the horizon length but polynomial in all other parameters, the set of {\it subgame value estimates} stabilizes for each player, at which point the product distribution across state-time pairs over the action profiles generated by continued local learning constitutes an EFCE for the stochastic game, thus circumventing the previous hardness result. 
In addition, we give a variant of PLL which removes the exponential dependence on horizon provided that a ``mixing'' assumption is satisfied. 

We also give an alternative approach which reaches an NFCCE in ``single-controller'' stochastic games, where only one player affects transitions (as studied in e.g.~\cite{singlecontroller}). Here, the controller uses a no-regret algorithm for adversarial MDPs with fixed transitions (\cite{Rosenberg2019OnlineSS}) while the followers use another variant of PLL. 
This approach converges in the black-box sense, where each agent has sublinear regret for the the uniform distribution over the entire history of strategies. We further show that the algorithms for general and fast-mixing stochastic games can be extended to satisfy sublinear regret bounds simultaneously for all agents if shared randomness is available by allowing agents to play according to the generated equilibrium after the initial algorithms terminate.
As building blocks for the analysis of our algorithms, we establish generalizations of known results for convergence of learning algorithms to correlated equilibria in normal-form games (e.g. \cite{blummansour}, to the case where reward feedback is noisy, which we call ``games with stochastic rewards'') and Bayesian games (removing the ``independent private value'' assumption in \cite{HST15}). Most proofs and some algorithmic details (such as exact constants) are deferred to Appendix A.

\subsection{Comparison with Related Work}

Most provably efficient algorithms for learning in stochastic games
target Nash equilibria in tractable special cases like zero-sum games, and often in infinite-horizon settings with discount factors or mixing guarantees
(\cite{rmax,CHF10,ZYL+18,ZYB18}).
When there are many players, we cannot afford to ``learn the game'' and use a model-based approach (e.g. \cite{rmax}), as explicitly representing even a single state will be intractable.
Closest to our setting is \cite{KMS00}, who give a centralized recursive algorithm that {\it computes} an EFCE in finite-horizon stochastic games for the case when the algorithm can sample many transitions and rewards at each state (which we cannot do in our ``repeated trajectories'' model); the runtime is exponential in both the horizon and the number of players, but does not depend on the number of states.
Correlated equilibria are also studied empirically by \cite{Greenwald2003CorrelatedQ}, and there is a large body of literature on general-sum multi-agent learning under other objectives or without convergence guarantees; for a recent overview of multi-agent reinforcement learning, see \cite{zhang2019multiagent}.

Finite-horizon stochastic games are somewhat related to extensive-form games, but are distinct in several important ways and in general are not directly comparable. 
In EFGs, only one player acts at each state, but partial information is allowed via ``infosets'', which can be used to simulate simultaneous actions (\cite{shoham_leyton-brown_2008,stengel,celli2020noregret}). Stochastic games with partial information have been considered in the literature (\cite{POSG}), but are considerably more difficult to solve (POMDPs, the single-player analog, are \textsf{PSPACE}-complete, see \cite{PT87}), and we will not consider them here.
EFGs typically enforce a tree structure on the infosets by the ``perfect recall'' assumption, whereas finite-horizon stochastic games allow for a DAG structure. As a result, our setting can allow for games with both a large depth and branching factor as long as the number of total states is not too large; EFGs are not as appropriate of a model when there are many paths to a given game state. 
Encoding a finite-horizon stochastic game as an EFG requires considering each path to a state independently, introducing a space blowup which is exponential in the horizon length, which renders existing methods for learning in EFGs impractical for our setting.

\section{Correlated Equilibria in Stochastic Games}
We begin with some background regarding no-(swap)-regret learning in the bandit feedback setting and connections to correlated equilibria in \Cref{subsec:prelims}. In \Cref{subsec:gwsr}, we introduce a preliminary model of a ``game with stochastic rewards'' and its corresponding definition of a correlated equilibrium. This serves as a building block for our formulation of stochastic games in \Cref{subsec:fhsg}. These game models may have unbounded description length; throughout, we treat them as oracles to which players submit actions simultaneously, then receive reward and state feedback. We assume instantaneous rewards are normalized to lie in $[0,1]$.

\subsection{Preliminaries}
\label{subsec:prelims}

\paragraph{Adversarial Bandits.}

In the {\it adversarial multi-armed bandit} problem,
the objective is to sequentially choose actions $a \in \A$
which minimize some notion of {\it regret}, where rewards at each step are chosen by an (adaptive) adversary. 
Let $N$ denote the cardinality of $\A$.
At each round, an algorithm commits to a distribution of actions $q^t \in \Delta(\A)$, which is observed by an adversary, who then chooses a reward vector $b^t \in [0,1]^N$. The algorithm then draws an action $a^t$ from $q^t$ and then observes the associated reward $b^t_{a^t}$.

Let $\Fsw$ denote the set of {\it swap functions} $\Fsw : \A \rightarrow \A$.
After $T$ rounds, the {\it swap-regret} of such an algorithm is given by
\begin{align*}
    \Reg^{\Fsw}(T) =&\; \max_{f \in \Fsw} \sum_{t=1}^T b^t_{f(a^t)} - \sum_{t=1}^T b^t_{a^t}. 
\end{align*}
Dividing by $T$ gives us the {\it average} swap regret; there are efficient algorithms for achieving sublinear swap-regret in this setting, which we refer to as {\it no-swap-regret} as average swap regret vanishes as $T$ grows.
\begin{prop}[\cite{blummansour}]
    \label{prop:srmab}
    There is an algorithm ($\srmab$) that, when used for $T$ rounds in the multi-armed bandit setting with adaptively chosen losses, has expected swap regret bounded by $O(N\sqrt{N T \log N})$. 
\end{prop}
This implies that after using $\srmab$ for $O(\frac{1}{\epsilon^2} N^3 \log N)$ rounds, the expected {average swap regret } is bounded by $\epsilon$. 
We will use $\B$ to denote the $\srmab$ algorithm and $B(\epsilon)$ to denote the number of rounds after which it has expected average swap regret at most $\epsilon$. We use it as a subroutine in our algorithms, but our results are not specific to its details, and can be adapted to use any no-swap-regret algorithm.  

\paragraph{Correlated Equilibria in Normal-Form Games.}

In a normal-form game with $M$ players and action space $\A = \times_{i \in [M]} \A_i$, each player $i$ selects an action $ a_i \in \A_i$ and receives a reward given by a utility function $u : \A \rightarrow [0,1]^M$ mapping action profiles to a vector of rewards. An $\epsilon$-{\it correlated equilibrium} for such a game is a {distribution} $D \in \Delta(\A)$ such that for all players $i$ and deterministic functions $f: \A_i \rightarrow \A_i$,
    $$\E_{a \sim D}[u_i(a_{i}; a_{-i})] \geq \E_{a \sim D}[u_i(f(a_i) ; a_{-i})] -\epsilon,$$
i.e.\ no player can benefit by more than $\epsilon$ in expectation by deviating from the distribution of ``recommended actions'' with any swap function. 
Repeated play in a normal form game converges to a correlated equilibrium if players use no-swap-regret algorithms.

\begin{prop}[\cite{blummansour}]
\label{prop:swce}
If all players in a game select actions using $\B$ for $B(\epsilon)$ rounds, the uniform distribution over the sequence of action profiles played thus far is an $\epsilon$-correlated equilibrium for the game, where the expectation is taken both with respect to the distribution of action profiles as well as the randomness of $\B$. 
\end{prop}

\subsection{Games with Stochastic Rewards}
\label{subsec:gwsr}

We define a {\it game with stochastic rewards} as a distribution over normal-form games, which is equivalent to a normal-form game where reward feedback can be noisy and arbitrarily correlated across players.

\begin{definition}[Games with Stochastic Rewards]
    A game with stochastic rewards $x = (\A, M, r, u)$ with $M$ players is given by a set of action profiles $\A = \times_{i \in [M]} \A_i$, a distribution over reward tensors $r \in \Delta(\Theta)$, and a utility function $u$, where the utilities $u : \A \times \Theta \rightarrow [0,1]^{M}$ depend on the realization of $\theta \sim r$. In a round of the game, players submit actions to $x$ simultaneously, $\theta$ is drawn independently from $r$, and each player observes only their utility $u_i(a, \theta)$.
\end{definition}

We assume that $\A_i = N$ for all agents. A correlated equilibrium for such a game is an action profile distribution where the regret bound holds with respect to the distribution over reward tensors.

\begin{definition}[Correlated Equilibria in Games with Stochastic Rewards]
    An $\epsilon$-correlated equilibrium for a game with stochastic rewards is a distribution $D \in \Delta(\A)$ such that for all players and all swap functions $f \in \Fsw$,
\begin{align*}
    \E_{a \sim D, \theta \sim r}[u_i(a_{i}; a_{-i}, \theta)] \geq \E_{a\sim D, \theta \sim r}[u_i(f(a_i) ; a_{-i}, \theta)] - \epsilon. 
\end{align*}
\end{definition}

Here, there is some expected reward tensor $\bar{\theta}$ where for every action profile $a$ and player $i$, $\bar{\theta}_{a,i} = \E_{\theta\sim r}[u_i(a_i; a_{-i}, \theta)]$, and a correlated equilibrium for such a game
is simply a correlated equilibrium for the game specified by $\bar{\theta}$.   

\subsection{Finite-Horizon Stochastic Games}
\label{subsec:fhsg}

Stochastic games resemble Markov decision processes, yet there are many players who act simultaneously at each state, and transition and reward dynamics depend on all players' actions.
In finite-horizon stochastic games, players begin at a state drawn from some initial distribution and play for a fixed period of steps, where a step consists of one set of simultaneous actions, followed by a transition to a new state and a reward for each agent. We allow both rewards and transitions to be probabilistic. A {\it trajectory} is the sequence of steps over the entire horizon length.

\begin{definition}[Finite-Horizon Stochastic Games]
    A {finite-horizon stochastic game} is given by a tuple $\M = (\X, \A, M, H, p_0, p, r, u)$, where:
\begin{itemize}
    \itemsep0em
    \item $M$ is the number of players,
    \item $\A = \times_{i \in [M]} \A_i$ is the action space ($\abs{\A_i} = N$ for each player),
    \item $H$ is the horizon length,
    \item $\X$ is the state space ($\abs{\X} = S$), 
    \item $p_0 \in \Delta(\X)$ is an initial distribution over states,
    \item $p : [H]\rightarrow \Delta(\Tau)$ is a function which defines distributions over transition functions $\tau \in \Tau : \A \times \X \rightarrow \X \cup \varnothing$,
    \item $r : \X \times H \rightarrow \Delta(\Theta)$ is a function which defines distributions of reward tensors, and 
    \item $u : \A \times \Theta \rightarrow [0,1]^M$ is a function which defines utilities for each player given an action profile and a reward tensor. 
\end{itemize}
For all $a, x, \tau$, we assume $\tau(a,x) = \varnothing$ if and only if $h = H$, where $\varnothing$ denotes termination of the episode.
\end{definition}

\paragraph{State Values for a Policy Profile.} We consider non-stationary policies of the form $\pi_i : \X \times [H] \rightarrow \A_i$ for each agent $i$, with $\pi_i \in \Pi_i$ and $\Pi = \times_{i \in [M]} \Pi_i$. 
For a policy profile $\pi$ we can recursively define a state value function $V^{\pi}_i: \X \times [H] \rightarrow [0,H]$ where:
\begin{align*}
    V^{\pi}_i(x, H) =&\; \E_{\theta \sim r(x, H)}[u_i(a_i; a_{-i}, \theta)] 
\end{align*}
where $a_i = \pi_i(x, h)$ for each agent $i$
and
\begin{align*}
    V^{\pi}_i(x, h) =&\; \E_{\theta, \tau}[u_i(a_i; a_{-i}, \theta) + V^{\pi}_i(\tau(a, x), h+1) ] 
\end{align*}
for $h \in \{1,\ldots, H-1\}$, where $\theta$ and $\tau$ are drawn from the appropriate distributions.

\paragraph{Counterfactual State Values.} 
We also define the {\it counterfactual} state value function for a player who deviates from a distribution over policy profiles. We consider two kinds of deviations: always playing a fixed policy $\psi_i \in \Pi_i$, or deviating from local recommendations using a ``swap function''.
Let $\Fsw_i : \A_i \times \X \times [H] \rightarrow \A_i$ be the set of swap functions for player $i$ that can depend on action, state, and episode step. 
We can recursively define a value function for $f \in \Fsw_i$ given a policy profile $\pi$:
\begin{align*}
    V^{\pi, f}_i(x, H) =&\; \E_{\theta \sim r(x, H)}[u_i(f(a_i, x, h); a_{-i}, \theta)] 
\end{align*}
and
\begin{align*}
    V^{\pi, f}_i(x, h) =&\; \E_{\theta, \tau}[u_i(f(a_i,x,h); a_{-i}, \theta) + V^{\pi, f}_i(\tau^*, h+1) ], 
\end{align*}
where $\tau^* = \tau(f(a_i, x, h), a_{-i}, x)$ and again $a_i = \pi_i(x, h)$ for each agent $i$.
Our notion of swap regret will be defined with respect to $\Fsw_i$, and our notion of a correlated equilibrium is a distribution over policy profiles $\pi = [\pi_i]_{i \in [M]}$. We can equivalently define $V^{\pi, \psi_i}_i(x, h)$ for $\psi_i \in \Pi_i$, omitting dependence on the actions recommended at each step.

We can adapt variants of correlated equilibria as considered for extensive-form games in e.g.\ \cite{stengel} or \cite{farina2019coarse} to stochastic games. Our definition of a normal-form correlated equilibrium says that no player can benefit substantially by committing to a fixed policy before seeing any recommendations, given knowledge of a policy profile distribution. For an extensive-form correlated equilibrium, recommendations are revealed to agents one step at a time, and they cannot benefit by deviating from these recommendations using a swap function. The EFCEs we consider will be a product distribution across state-time pairs, and so we restrict to considering deviations based only on the current recommendation---recommendations at previous steps provide no additional information about opponent recommendations at any other step.

\begin{definition}[Normal-Form Coarse Correlated Equilibria for Stochastic Games]

We say that 
a policy profile distribution $D \in \Delta(\Pi)$
is an $\epsilon$-approximate normal-form coarse correlated equilibrium (or $\epsilon$-NFCCE) for a finite-horizon stochastic game if for all agents $i$ and all $\psi_i \in \Pi_i$:
\begin{align*}
    \E_{x \sim p_0, \pi \sim D} [V^{\pi}_i(x, 1)] \geq&\; \E_{x\sim p_0, \pi \sim D} [V^{\pi, \psi}_i(x, 1)] - \epsilon H .
\end{align*}
\end{definition}

\begin{definition}[Extensive-Form Correlated Equilibria for Stochastic Games]
    \label{def:ce-fhsg}
We say that 
a policy profile distribution $D \in \Delta(\Pi)$
is an $\epsilon$-approximate extensive-form correlated equilibrium (or $\epsilon$-EFCE) for a finite-horizon stochastic game if for all agents $i$ and all $f \in \Fsw_i$:
\begin{align*}
    \E_{x \sim p_0, \pi \sim D} [V^{\pi}_i(x, 1)] \geq&\; \E_{x\sim p_0, \pi \sim D} [V^{\pi, f}_i(x, 1)] - \epsilon H .
\end{align*}

\end{definition}
Just as in the case for extensive-form games, EFCEs provide stronger guarantees than NFCCEs.
\begin{theorem}
For a finite-horizon stochastic game, the set of $\epsilon$-EFCEs is contained in the set of $\epsilon$-NFCCEs for all $\epsilon \geq 0$.
\end{theorem}
\begin{proof}
Any fixed policy $\psi$ can be encoded with the swap function $f(a_i, x, h) = \psi(x,h)$ for all $a_i$, $x$, and $h$, and so any $\epsilon$-EFCE is also an $\epsilon$-NFCCE.
\end{proof}

These definitions bound the average per-step regret by $\epsilon$ for each agent under the appropriate class of deviations. We are interested in when, and how quickly, players can converge to such an equilibrium by repeatedly playing the game.

\section{Hardness of Learning in Adversarial MDPs}
\label{sec:adv-mdp}
The first thing that one might hope for in the setting of multi-player finite-horizon stochastic games is the existence of an algorithm which minimizes the appropriate notion of regret for each agent, and can be used as a {\it black box} to reach a correlated equilibrium. This is the form of \cite{celli2020noregret}, who give an algorithm with sublinear {\it trigger regret}, which corresponds to the definition of an EFCE and thus results in efficient convergence of the sequence of policies played to an EFCE when all agents use the algorithm.

The appropriate problem for modeling repeated play in finite-horizon stochastic games against arbitrary opponents is the ``adversarial MDP problem'', where an agent is faced with a set of finite-horizon MDPs, each with a different reward and transition function. A commonly studied objective is to minimize regret against the best fixed policy, and such an algorithm with sublinear regret would converge to an NFCCE for a stochastic game. This was shown to be as hard as the ``learning parities with noise'' problem (which is not known to be \textsf{NP}-hard) by \cite{NIPS2013_4f284803} when $H = \Theta(S)$.  
We show that this is indeed \textsf{NP}-hard even when the horizon is only 3.

\begin{theorem}\label{thm:hard-adv-mdp}
Assuming $\textsf{\textup{NP}} \not\subseteq \textsf{\textup{BPP}}$, there is no algorithm with polynomial time per-round computation which has  $O(T^{1 - \delta} \cdot \textup{poly}(S))$ regret algorithm for the adversarial MDP problem with $H \geq 3$, for any $\delta > 0$.
\end{theorem}

We prove this by considering an offline version of the adversarial MDP problem, where the goal is to find a single non-stationary policy which does well across a set of MDPs with differing reward and transition functions. We show that this is as hard as \textsf{MAX-3-SAT}, and use the online-to-batch reduction from \cite{CCG04} to show hardness of the online problem. This suggests we should not expect a black-box reduction to finding even an NFCCE, even when the horizon is quite short.

\begin{corollary}\label{cor:hard-nr-sg}
Assuming $\textsf{\textup{NP}} \not\subseteq \textsf{\textup{BPP}}$, any decentralized learning algorithm which converges in polynomial time to an approximate (coarse) correlated equilibrium for a stochastic game with horizon $H \geq 4$ when used by all players must have regret $\omega(T^{1 - \delta} \cdot \textup{poly}(S))$, for all $\delta > 0$, against arbitrary opponents.
\end{corollary}

Despite this, we give algorithms which converge to correlated equilibria which are not black-box, i.e.\ they explicitly make use of the fact that all players are using the same algorithm. Without any assumptions on transitions, the runtime of our primary algorithm, PLL, is polynomial in all parameters except the horizon, where dependence is exponential in the worst case, allowing us to overcome the barrier we show for black-box reductions.

\section{Learning in Stochastic Games via Repeated Trajectories}
\label{sec:fhsg}

The main idea behind our algorithm is for each agent to locally perform no-swap-regret learning at each state-time pair, augmenting their observed rewards with estimates of the ``values'' for states they transition to. 
We first give an extension of the convergence theorem for bandit learning in normal-form games from \cite{blummansour} to ``games with stochastic rewards'', which makes use of $\B$ with additional modifications in order to handle stochasticity and obtain high-probability bounds for both regret and value estimates. We then give an ``offline'' centralized algorithm, BILL, which uses this subroutine to {\it compute} an EFCE for a stochastic game, given the ability to sample rewards and transitions for each state. Our algorithm PLL can be viewed as simulating BILL in a decentralized manner when agents play repeated trajectories of the game. The sense in which PLL converges is different from e.g.\ \cite{blummansour}; rather than taking the uniform distribution over the history of policies, we consider the product distribution of a truncated history of action profiles at each state-time pair. 
We can improve the speed of convergence for PLL when a ``fast-mixing'' assumption is satisfied, a common tool in the analysis of reinforcement learning algorithms.

\subsection{Learning in Games with Stochastic Rewards}

Recall that we define correlated equilibria for stochastic games with respect to the average reward tensor $\bar{\theta}$. When agents all use a no-swap regret algorithm (such as $\B$) to play such a game repeatedly, the immediate regret bound holds with respect to the realized sequence of reward tensors. We can extend this bound to hold with respect to $\bar{\theta}$ by viewing the ``error'' of each swap function for a player (their reward from sampled sequence of reward tensors $\theta$ versus the average tensor $\bar{\theta}$) as a martingale which does not deviate too far from its expectation. Depending on the relationship between $\epsilon$ and $N$, we may need to run $\B$ for slightly longer than $B(\epsilon)$ in order to apply our martingale analysis, but only by at most a factor of $O(\log(1/\epsilon))$. We let $B(\epsilon,N)$ denote this extended runtime as a function of $\epsilon$ and $N$.

\begin{theorem}
\label{thm:gwsr-ce}
    When players in a game with stochastic rewards $x$ select actions using $\B$ for $T \geq B(\epsilon/4, N)$ rounds, 
    the sequence of action profiles is an $\epsilon$-correlated equilibrium for the game, where the expectation is taken with respect to the tensor 
    distribution as well as $\B$.
\end{theorem}
    
The proof is given in Appendix A.2. By running $\B$ several times, we can boost the expected regret bound for each player to hold with high probability over the randomness of  while simultaneously obtaining accurate estimates of the {\it value} of this process for each player; we use this form of the result in the analysis for later algorithms.


\begin{corollary}
\label{cor:gwsr-ce-whp}
  When all agents in a game with stochastic rewards $x$ play according to $\B$ for at least $\frac{2 \log\parens{5M/\delta} }{\eta^2} \cdot B(\epsilon/8, N)$ rounds, simultaneously restarting $\B$ every $B(\epsilon/8, N)$ rounds, the resulting sequence of actions is an $(\epsilon/2 + \eta/2)$-correlated equilibrium for $x$ with probability at least $1 - \delta/5$.
 
 Further, let $V_i^{\B}(x) = \E_{\B,r}\left[\frac{1}{T}\sum_{t=1}^T u_i(a_i^t; a^t_{-i}, \theta^t) \right]$ and let $\hat{V}_i^{\B}(x)$ be the average utility received by player $i$ over all rounds.  
 With probability at least $1 - 2 \delta/5$, $\abs{V_i^{\B}(x) - \hat{V}_i^{\B}(x) } \leq \eta/2$ simultaneously for all players.

 Additionally, the computed estimate is within $\eta$ of player $i$'s expected average reward for playing the game according to the resulting policy distribution with probability at least $1 - 2\delta/5$. 
\end{corollary}

An extension of this method to Bayesian games is presented in Appendix A.3, which we make use of in analyzing Algorithm 4 (\Cref{thm:sc-fhsg-ce}).

\subsection{Subgame Value Estimates} 
We define a notion
of the {\it subgame value} for an agent at a state-step pair $(x,h)$
in a stochastic game, similiar to that in \Cref{def:ce-fhsg},
which is specified with respect to a learning algorithm $\B$. 
Henceforth we will refer to $(x,h)$ simply as a {\it pair}. 
We will define this recursively.
Note that a pair $(x, H)$ in a finite-horizon stochastic game is equivalent to a game with stochastic rewards, as all action profiles result in termination of the episode. 
If 

all agents 
play according to private copies of a bandit algorithm $\B$ for $T$ rounds,
the average reward for each agent over the period can be viewed as a random variable, where the  
expected value for agent $i$ is
given by:
\begin{align*}
    V_i^{\B}(x, H) =&\; \E_{\{a^t \} \sim(\B)_{i \in [M]}, \theta \sim r(x, H)} \left[\frac{1}{T} \sum_{t=1}^T u_i(a_i^t ; a_{-i}, \theta) \right].
\end{align*}
This will be in $[0,1]$ for all agents.
We can also view other pairs $(x, h)$ as games with stochastic rewards, where the immediate reward for an agent is augmented with their value of the state they transition to.
Values of states in steps prior to $H$ will represent the expected reward of an agent in the remainder of the episode when all agents play at each state according to $\B$ at each pair, augmenting their immediate payoffs at a pair with the value of the pair they transition to.  
Suppose $V^{\B}_i(x', h')$ is defined for all $x'\in \X$ and for all $h' > h$. Then,
\begin{align*}
    V^{\B}_i(x, h) =&\; \E \left[\frac{1}{T} \sum_{t=1}^T u_i(a_i^t ; a_{-i}, \theta) + V_i^{\B}(\tau(a, x), h+1)\right],
\end{align*}
where the expectation is taken over the randomness of each copy of $\B$ as well as sampled reward tensors and transition functions.
These will be in $[0,H-h+1]$, but throughout, we will assume that rewards are scaled to $[0,1]$ before being given to $\B$. 
These subgame values we have defined represent the utility which an agent can obtain in expectation if they use a copy of $\B$ at each state and know all downstream subgame values. 
Subgame values can be equivalently defined using downstream value {\it estimates} $\hat{V}_i$, and we obtain such estimates from Corollary 3.1 which are accurate with high probability.

\subsection{An Efficient Offline Algorithm}

If we are not constrained to learning online through entire trajectories, and can sample reward tensors and transition functions from any state-step pair (as oracles with constant-time query access), 
there is a straightforward offline algorithm for computing an EFCE which is a product distribution across pairs.

\paragraph{Algorithm 1: Backward-Inductive Local Learning.}
\begin{itemize}
    \item Use a copy of a bandit algorithm $\B$ for each player
        to compute an approximate correlated equilibrium and value estimates $\hat{V}_i(x, H)$ for each player and pair, as in Corollary 3.1.  
    \item By backward induction, compute approximate equilibria and value estimates for each pair $(x, h)$ in the same manner, augmenting players' rewards at state $x$ with value estimates for $(x', h+1)$, where $x'$ is the visited state in step $h+1$ for that round. 
    \item Return the product distribution of computed sequences of action profiles across all pairs.
\end{itemize}
\begin{theorem}\label{thm:bill}
BILL computes an $\epsilon$-EFCE in $\tilde{O}(\poly(M, N, S,H, 1/\epsilon))$ time.
\end{theorem}
The proof is quite similar to the error propagation analysis for \Cref{thm:fhsg-ce}.

\subsection{Parallel Local Learning}
\label{sec:inductive-bandits}

PLL essentially simulates BILL in a decentralized manner when used by all agents by computing estimates of subgame values for each agent over a series of {\it epochs},
which are batches of many trajectories of the game,
by using a no-swap-regret algorithm $\B$ at each state. We say that a state is {\it locked} when it is visited enough to obtain an accurate value estimate, and estimates are reset whenever a value estimate for a downstream step is updated.
We terminate once an epoch elapses where no new states are locked.

A key point of difficulty here is that
visitation probabilities may shift drastically when value estimates change; the sequence of actions taken when all players use $\B$ at a pair may be quite sensitive to small changes in rewards even for just one player. 
We show that the number of epochs before termination is at most exponential in $H$, at which point
the distribution over action profiles at each pair truncated at the last reset constitutes an approximate correlated equilibrium for the subgame at that pair (given downstream values) almost surely, with the exception of pairs which are visited infrequently under the final value estimates.
We then show how regret bounds compose to give an $\epsilon$-EFCE for the entire finite-horizon game when considering action profiles sampled independently across pairs from the aforementioned distributions.

\paragraph{Algorithm 2: Parallel Local Learning.}
Initialize $\hat{V}_i^{\B}(x, h) = H-h+1$ for each pair $(x, h)$, as well as a visit counter $c(x,h)$ for each pair set to 0. 
Let $W = \tilde{\Theta}\parens{\frac{S^4 H^7}{\epsilon^2}}$ and $L = \Theta\parens{\frac{S^2 H^4 W B}{\epsilon^2}}$.
Initialize a copy of $\B$ at each pair, specified to run for $B = B(\frac{\epsilon}{16H}, N)$ steps.
 Until termination, run the following procedure for each epoch:
\begin{itemize}
\item Run for $L$ trajectories, using $\B$ at each pair, counting rounds and updating actions for a copy of $\B$ only when the corresponding pair is visited. Record rewards as the sum of the observed reward as well as the current value estimate for the {\it next} pair visited in that trajectory, scaled to [0,1]. 
\item Consider the last step  $h \in [H]$ where an unlocked pair's counter crossed $\frac{16H^2WB}{\epsilon}$ in the epoch. 
Lock all unlocked states at this step with appropriate estimates which were previously unlocked, compute value estimates $\hat{V}_i^{\B}(x, h)$ as the average reward over the corresponding
$\frac{16H^2WB}{\epsilon}$ visits,
then reset all copies of $\B$, counters, and value estimates at {\it earlier} pairs $(h' < h)$.
\item Terminate if no pair's counter crosses $\frac{16H^2WB}{\epsilon}$ in the epoch.
\end{itemize}
Note that when all players use this algorithm, locking and unlocking is synchronized across players.
The action profile distributions for each pair {\it after they are last unlocked} converge to an approximate EFCE for the game, when we consider action profiles sampled independently for each pair, with a running time at most exponential in the horizon and polynomial in all other parameters.    
\begin{theorem}\label{thm:fhsg-ce}
    PLL terminates after at most $(S+1)^H+1$ epochs.
    After termination, for each pair (x, h), consider the uniform distribution over action profiles $D(x,h)$ played since that pair was last reset. 
    Let $D$ be the distribution over policy profiles where the action profile for each pair $(x,h)$ is sampled independently from $D(x, h)$. 
With probability at least $1 - \delta$, $D$ is an $\epsilon$-EFCE for the game.
\end{theorem}

A key step in the analysis of PLL is to bound the number of times that value estimates can change, thus bounding the number of required epochs before estimates stabilize. 
\begin{lemma}\label{lemma:pll-epochs}
    The algorithm runs for at least $H$ epochs, and at most $(S+1)^H+1$ epochs. 
\end{lemma}
\begin{proof}
    All pairs start unlocked, and some pair in each step is visited at least $\frac{16H^2 WB}{\epsilon}$ per epoch by pigeonhole, so the algorithm will not terminate unless there is a locked pair for every step. 
States are only moved from unlocked to locked at one step per epoch, and so there must be at most $H$ epochs to lock some pair in all steps.

We can bound the number of epochs by bounding the number of times a pair at some step can become locked. Observe that a locked pair at step $H$ will only become locked in one epoch and will never become unlocked afterwards. A pair at step $H-1$ will become locked in at most $S$ epochs, as it will only become locked after at least one pair at step $H$ is locked, and then can be unlocked at most $S-1$ times for the remaining unlocked pairs at step $H$.   
In general, the number of epochs in which a state can become locked is bounded by the number of epochs in which a downstream state can become locked. 
Let $g(h)$ denote this bound on the number of epochs in which a pair at step $h$ can be locked, which is given by:
\begin{align*}
    g(h) =&\; \sum_{i = h+1}^H S g(i) \\ 
    =&\; Sg(h+1) + \sum_{i = h+2}^H S g(i) \\ 
    =&\; (S+1)g(h+1) \\ 
    =&\; (S+1)^{H-h} g(H) \\
    =&\; (S+1)^{H-h},
\end{align*}
as $g(H) = 1$.
The total number of epochs before termination is then bounded by
\begin{align*}
    1 + \sum_{i=1}^H Sg(i) = g(0) + 1 = (S+1)^H + 1, 
\end{align*}
accounting for the last epoch in which no states are locked. 
\end{proof}

Given this, much of the remainder of the analysis is to analyze the propagation of estimation error and regret terms to give an explicit bound on the regret after estimates have stabilized.

\subsection{Efficient Learning in Fast-Mixing Stochastic Games}
\label{subsec:fast-mixing}

PLL generates an EFCE in polynomial time only when $H$ is a constant. 
For ``fast-mixing'' games we give a related algorithm, FastPLL, 
which converges to an $\epsilon$-EFCE in finite-horizon stochastic games which are {\it $\gamma$-fast-mixing} in time $\tilO(\poly(S, N, H, 1/\epsilon, 1/\gamma)$. 
We will say that a finite-horizon stochastic game is $\gamma$-fast-mixing if all pairs
are visited with probability at least $\gamma$ in a trajectory
when each agent selects a policy uniformly at random,
i.e.\ for each $(x, h)$:
\begin{align*}
    \Pr_{\pi \sim \Pi, z}\left[ x \text{ is visited at step } h \right] \geq&\; \gamma.
\end{align*}

Unlike the previous algorithm, the fast-mixing assumption allows us to avoid unlocking states once they are locked, as we can ensure sufficient visitation with high probability. As a result, we show that polynomial time convergence to a correlated equilibrium is possible after only $H$ epochs.

\paragraph{Algorithm 3: Fast PLL.}
Let $B = B\parens{\frac{ \epsilon }{8H}, N}$, and let the epoch length (in trajectories) be given by 
$L = \tilde{\Theta}\parens{\frac{BH^4}{\gamma^{1.5} \epsilon^2}}$.
Run $H$ epochs, one corresponding to each step (beginning with step $H$) as follows:

\begin{itemize}
    \itemsep0em
    \item \emph{Epoch for Step $h$:} Use a copy of $\B$ to select actions at each pair $(x, h)$, augmenting rewards with computed values for pairs $(x', h+1)$ transitioned to for the next step (if $h < H$). At the end of the epoch, let $\hat{V}_i^{\B}(x, h)$ be the average reward received from all completed runs of $\B$.
    \item \emph{Upstream ($h' < h$):} Select actions uniformly at random for each pair. 
    \item \emph{Downstream ($h' > h$):} Use $\B$ at each signal as in the epoch for step $h'$, 
        augmenting rewards with value estimates for pairs transitioned to.
      Restart $\B$ after every $B$ rounds in which it is used, which can include rounds from a prior epoch. 
\end{itemize}

The notion of convergence here is the same as that for PLL.

\begin{theorem}\label{thm:fast-mixing-ce}
    After Algorithm 3 terminates, for each pair $(x, h)$, consider the uniform distribution over action profiles $D(x,h)$ played since epoch $H-h+1$ began. 
    Let $D$ be the distribution over policy profiles where the action profile for each pair $(x,h)$ is sampled independently from $D(x, h)$. 
With probability at least $1 - \delta$, $D$ is an $\epsilon$-EFCE for the game.
\end{theorem}

\section{When Can We Get Simultaneous No-Regret?}
\label{sec:faster}

While PLL gives us a way to generate an EFCE, as well as find stable value estimates for all pairs and players, it is not itself a no-regret algorithm.
For single-controller stochastic games, where only one player (the controller) affects transitions, we show that an NFCCE can be reached without shared randomness when the controller uses an algorithm for adversarial MDPs with fixed transitions and each follower uses $\B$ repeatedly in parallel across each pair. Further, both PLL and FastPLL can again be extended to simultaneous no-swap-regret algorithms in the case where {\it shared randomness} is available for all players.

\subsection{Efficient Learning in Single-Controller Stochastic Games}
\label{sec:single-control}

When only one player affects transitions, their problem is equivalent to an adversarial MDP with fixed transitions. The Shifted Bandits U-CO-REPS algorithm from \cite{Rosenberg2019OnlineSS} obtains sublinear regret in finite-horizon adversarial MDPs of this form with only bandit feedback and when the transition function is unknown. 
We show that running Shifted Bandits UC-O-REPS for the ``controller'' bounds their appropriate notion of regret against arbitrary ``followers''.
The learning problem for the followers can be viewed as a set of {\it Bayesian} games with shifting signal distributions.
In Appendix A.3 we give an extension of our analysis of games with stochastic rewards to Bayesian games, which generalizes the convergence result of \cite{HST15} to remove the ``independent private value'' assumption, and which we can use to prove a regret bound for a modification of $\B$ (which we call a ``parallel bandit'' algorithm, denoted $\B_S$) against arbitrary opponents. The regret bound holds even when the ``signal distribution'' for the Bayesian game shifts over time, and the followers will use a copy of $\B_S$ for each time-step. 
As such, all agents can efficiently reach an NFCCE by black-box regret minimization.
   
\paragraph{Algorithm 4: S.B.\ U-CO-REPS + P.B.}

Let $B_L(\epsilon)$ be the time after which S.B.\ U-CO-REPS has per-step regret $\epsilon$, which is $\poly(H,S,N,1/\epsilon)$,
 and let $B_F(\epsilon) = B(\epsilon/S, N)$.
Run for $T = \frac{8\log(M/\delta)}{\epsilon^2} \cdot \max \parens{B_L(\epsilon/8), B_F(\epsilon/8)} $ total trajectories, where each player acts as follows:
\begin{itemize}
    \item \emph{Controller:} Select policies for each trajectory using S.B. U-CO-REPS, restarting every $B_L(\epsilon / 8)$ trajectories.
    \item \emph{Followers:} Select policies using a copy of $\B_S$ for Bayesian games at each step, counting only immediate rewards, and restarting every $B_F({\epsilon}/{8})$ trajectories.  
\end{itemize}
This specifies a policy for each player prior to the start of each trajectory, and this sequence of policies will converge to an approximate NFCCE. 

\begin{theorem}\label{thm:sc-fhsg-ce}
With probability at least $1 - \delta$, the uniform distribution over the sequence of policy profiles played by Algorithm 4 is an $\epsilon$-NFCCE for the game.
\end{theorem}

\subsection{Simultaneous No-Swap-Regret with Shared Randomness}
\label{subsec:shared-random}
If players have access to shared randomness at each step, they can play according to the equilibrium generated by PLL or FastPLL in future rounds without any explicit communication. 
The total regret bound is sublinear in $T$ when the ``target average reget'' for the PLL (or FastPLL) portion is appropriately calibrated so that the any regret incurred at the beginning does not overwhelm the average regret for the entire sequence of play.

\paragraph{Algorithm 5: PLL with Shared Randomness (PLL-SR).}
Let $\epsilon_1 = \tilde{\Theta}\parens{\sqrt[7]{\frac{N^3 S^{O(H)}}{T}}}$ and $\epsilon_2 = \tilde{\Theta} \parens{\sqrt[5]{\frac{N^3 H^4 \gamma^{2/3}}{T}}}$.
\begin{itemize}
    \item Run PLL, specified for an $\epsilon_1$-EFCE, until termination, or FastPLL for an $\epsilon_2$-EFCE.
    \item At each step after termination, each player receives the same uniform random number $w \in [W^*]$ and plays the $w$th action of the final high-probability local CE sequence (from \Cref{cor:gwsr-ce-whp}), where $W^*$
    is the appropriate length of the sequence for PLL or FastPLL.
\end{itemize}

\begin{theorem}\label{thm:no-swap-shared-random}
With respect to $\Fsw$,
PLL-SR has regret $\tilde{O}(T^{\frac{6}{7}})$ and FastPLL-SR has regret $\tilde{O}(T^{\frac{4}{5}})$.
\end{theorem}

\begin{acknowledgements} 
    We thank Christos Papadimitriou and Tim Roughgarden for helpful feedback and suggestions throughout this work, and Kiran Vodrahalli and Utkarsh Patange for illuminating discussions regarding the hardness result. 
\end{acknowledgements}

\clearpage
\bibliography{ref}

\begin{thebibliography}{35}
\providecommand{\natexlab}[1]{#1}
\providecommand{\url}[1]{\texttt{#1}}
\expandafter\ifx\csname urlstyle\endcsname\relax
  \providecommand{\doi}[1]{doi: #1}\else
  \providecommand{\doi}{doi: \begingroup \urlstyle{rm}\Url}\fi

\bibitem[Abbasi-Yadkori et~al.(2013)Abbasi-Yadkori, Bartlett, Kanade, Seldin,
  and Szepesvari]{NIPS2013_4f284803}
Yasin Abbasi-Yadkori, Peter~L Bartlett, Varun Kanade, Yevgeny Seldin, and Csaba
  Szepesvari.
\newblock Online learning in markov decision processes with adversarially
  chosen transition probability distributions.
\newblock In C.~J.~C. Burges, L.~Bottou, M.~Welling, Z.~Ghahramani, and K.~Q.
  Weinberger, editors, \emph{Advances in Neural Information Processing
  Systems}, volume~26, pages 2508--2516. Curran Associates, Inc., 2013.

\bibitem[Aumann(1974)]{AUMANN197467}
Robert~J. Aumann.
\newblock Subjectivity and correlation in randomized strategies.
\newblock \emph{Journal of Mathematical Economics}, 1\penalty0 (1):\penalty0
  67--96, 1974.
\newblock ISSN 0304-4068.
\newblock \doi{https://doi.org/10.1016/0304-4068(74)90037-8}.

\bibitem[Bergemann and Morris(2016)]{doi:10.3982/TE1808}
Dirk Bergemann and Stephen Morris.
\newblock Bayes correlated equilibrium and the comparison of information
  structures in games.
\newblock \emph{Theoretical Economics}, 11\penalty0 (2):\penalty0 487--522,
  2016.
\newblock \doi{10.3982/TE1808}.

\bibitem[Blum and Mansour(2004)]{blummansour}
Avrim Blum and Yishay Mansour.
\newblock From external to internal regret.
\newblock volume~8, 05 2004.
\newblock \doi{10.1007/11503415_42}.

\bibitem[Brafman and Tennenholtz(2001)]{rmax}
Ronen Brafman and Moshe Tennenholtz.
\newblock R-max - a general polynomial time algorithm for near-optimal
  reinforcement learning.
\newblock volume~3, pages 953--958, 01 2001.
\newblock \doi{10.1162/153244303765208377}.

\bibitem[Celli et~al.(2020)Celli, Marchesi, Farina, and
  Gatti]{celli2020noregret}
Andrea Celli, Alberto Marchesi, Gabriele Farina, and Nicola Gatti.
\newblock No-regret learning dynamics for extensive-form correlated and coarse
  correlated equilibria.
\newblock \emph{CoRR}, abs/2004.00603, 2020.
\newblock URL \url{https://arxiv.org/abs/2004.00603}.

\bibitem[{Cesa-Bianchi} et~al.(2004){Cesa-Bianchi}, {Conconi}, and
  {Gentile}]{CCG04}
N.~{Cesa-Bianchi}, A.~{Conconi}, and C.~{Gentile}.
\newblock On the generalization ability of on-line learning algorithms.
\newblock \emph{IEEE Transactions on Information Theory}, 50\penalty0
  (9):\penalty0 2050--2057, 2004.
\newblock \doi{10.1109/TIT.2004.833339}.

\bibitem[{Chang} et~al.(2010){Chang}, {Hu}, {Fu}, and {Marcus}]{CHF10}
H.~S. {Chang}, J.~{Hu}, M.~C. {Fu}, and S.~I. {Marcus}.
\newblock Adaptive adversarial multi-armed bandit approach to two-person
  zero-sum markov games.
\newblock \emph{IEEE Transactions on Automatic Control}, 55\penalty0
  (2):\penalty0 463--468, 2010.

\bibitem[Chen and Peng(2020)]{chen2020hedging}
Xi~Chen and Binghui Peng.
\newblock Hedging in games: Faster convergence of external and swap regrets.
\newblock \emph{CoRR}, abs/2006.04953, 2020.
\newblock URL \url{https://arxiv.org/abs/2006.04953}.

\bibitem[Chen et~al.(2007)Chen, Deng, and Teng]{CDT07}
Xi~Chen, Xiaotie Deng, and Shang{-}Hua Teng.
\newblock Settling the complexity of computing two-player nash equilibria.
\newblock \emph{CoRR}, abs/0704.1678, 2007.

\bibitem[Daskalakis et~al.(2006)Daskalakis, Goldberg, and Papadimitriou]{DGP06}
Constantinos Daskalakis, Paul~W. Goldberg, and Christos~H. Papadimitriou.
\newblock The complexity of computing a nash equilibrium.
\newblock In \emph{Proceedings of the Thirty-Eighth Annual ACM Symposium on
  Theory of Computing}, STOC ’06, page 71–78, New York, NY, USA, 2006.
  Association for Computing Machinery.
\newblock ISBN 1595931341.
\newblock \doi{10.1145/1132516.1132527}.

\bibitem[Even-Dar et~al.(2004)Even-Dar, Kakade, and Mansour]{ExpertsMDP04}
Eyal Even-Dar, Sham~M. Kakade, and Yishay Mansour.
\newblock Experts in a markov decision process.
\newblock In \emph{Proceedings of the 17th International Conference on Neural
  Information Processing Systems}, NIPS'04, page 401–408, Cambridge, MA, USA,
  2004. MIT Press.

\bibitem[Farina et~al.(2019)Farina, Bianchi, and Sandholm]{farina2019coarse}
Gabriele Farina, Tommaso Bianchi, and Tuomas Sandholm.
\newblock Coarse correlation in extensive-form games.
\newblock \emph{CoRR}, abs/1908.09893, 2019.
\newblock URL \url{http://arxiv.org/abs/1908.09893}.

\bibitem[Filar and Raghavan(1984)]{singlecontroller}
Jerzy~A. Filar and T.~E.~S. Raghavan.
\newblock A matrix game solution of the single-controller stochastic game.
\newblock \emph{Mathematics of Operations Research}, 9\penalty0 (3):\penalty0
  356--362, 1984.
\newblock \doi{10.1287/moor.9.3.356}.

\bibitem[Forges(1993)]{RePEc:cor:louvco:1993009}
Francoise Forges.
\newblock Five legitimate definitions of correlated equilibrium in games with
  incomplete information.
\newblock CORE Discussion Papers 1993009, Université catholique de Louvain,
  Center for Operations Research and Econometrics (CORE), 1993.

\bibitem[Foster and Vohra(1997)]{FOSTER199740}
Dean~P. Foster and Rakesh~V. Vohra.
\newblock Calibrated learning and correlated equilibrium.
\newblock \emph{Games and Economic Behavior}, 21\penalty0 (1):\penalty0 40--55,
  1997.
\newblock ISSN 0899-8256.
\newblock \doi{https://doi.org/10.1006/game.1997.0595}.

\bibitem[Greenwald and Hall(2003)]{Greenwald2003CorrelatedQ}
Amy Greenwald and Keith Hall.
\newblock Correlated q-learning.
\newblock In \emph{ICML}, 2003.

\bibitem[Hansen et~al.(2004)Hansen, Bernstein, and Zilberstein]{POSG}
Eric~A. Hansen, Daniel~S. Bernstein, and Shlomo Zilberstein.
\newblock Dynamic programming for partially observable stochastic games.
\newblock In \emph{Proceedings of the 19th National Conference on Artifical
  Intelligence}, AAAI'04, page 709–715. AAAI Press, 2004.
\newblock ISBN 0262511835.

\bibitem[Hart and Mas-Colell(2000)]{hartmascolell}
Sergiu Hart and Andreu Mas-Colell.
\newblock A simple adaptive procedure leading to correlated equilibrium.
\newblock \emph{Econometrica}, 68:\penalty0 1127--1150, 09 2000.
\newblock \doi{10.1111/1468-0262.00153}.

\bibitem[Hartline et~al.(2015)Hartline, Syrgkanis, and Tardos]{HST15}
Jason Hartline, Vasilis Syrgkanis, and \'{E}va Tardos.
\newblock No-regret learning in bayesian games.
\newblock In \emph{Proceedings of the 28th International Conference on Neural
  Information Processing Systems - Volume 2}, NIPS’15, page 3061–3069,
  Cambridge, MA, USA, 2015. MIT Press.

\bibitem[H\r{a}stad(1997)]{Has97}
Johan H\r{a}stad.
\newblock Some optimal inapproximability results.
\newblock In \emph{Proceedings of the Twenty-Ninth Annual ACM Symposium on
  Theory of Computing}, STOC '97, page 1–10, New York, NY, USA, 1997.
  Association for Computing Machinery.
\newblock ISBN 0897918886.
\newblock \doi{10.1145/258533.258536}.
\newblock URL \url{https://doi.org/10.1145/258533.258536}.

\bibitem[Huang and von Stengel(2008)]{efcepoly}
Wan Huang and Bernhard von Stengel.
\newblock Computing an extensive-form correlated equilibrium in polynomial
  time.
\newblock pages 506--513, 12 2008.
\newblock \doi{10.1007/978-3-540-92185-1_56}.

\bibitem[Kearns et~al.(2000)Kearns, Mansour, and Singh]{KMS00}
Michael~J. Kearns, Yishay Mansour, and Satinder~P. Singh.
\newblock Fast planning in stochastic games.
\newblock \emph{UAI}, 2000.

\bibitem[Littman(1994)]{Littman1994MarkovGA}
Michael~L. Littman.
\newblock Markov games as a framework for multi-agent reinforcement learning.
\newblock In \emph{ICML}, 1994.

\bibitem[Nash(1950)]{Nash48}
John~F. Nash.
\newblock Equilibrium points in n-person games.
\newblock \emph{Proceedings of the National Academy of Sciences}, 36\penalty0
  (1):\penalty0 48--49, 1950.
\newblock ISSN 0027-8424.
\newblock \doi{10.1073/pnas.36.1.48}.

\bibitem[Nisan et~al.(2007)Nisan, Roughgarden, Tardos, and
  Vazirani]{nisan_roughgarden_tardos_vazirani_2007}
Noam Nisan, Tim Roughgarden, Eva Tardos, and Vijay~V. Vazirani.
\newblock \emph{Algorithmic Game Theory}.
\newblock Cambridge University Press, 2007.
\newblock \doi{10.1017/CBO9780511800481}.

\bibitem[Papadimitriou and Tsitsiklis(1987)]{PT87}
Christos~H. Papadimitriou and John~N. Tsitsiklis.
\newblock The complexity of markov decision processes.
\newblock \emph{Mathematics of Operations Research}, 12\penalty0 (3):\penalty0
  441--450, 1987.
\newblock ISSN 0364765X, 15265471.

\bibitem[Rosenberg and Mansour(2019)]{Rosenberg2019OnlineSS}
Aviv Rosenberg and Yishay Mansour.
\newblock Online stochastic shortest path with bandit feedback and unknown
  transition function.
\newblock In \emph{NeurIPS}, 2019.

\bibitem[Shoham and Leyton-Brown(2008)]{shoham_leyton-brown_2008}
Yoav Shoham and Kevin Leyton-Brown.
\newblock \emph{Multiagent Systems: Algorithmic, Game-Theoretic, and Logical
  Foundations}.
\newblock Cambridge University Press, 2008.
\newblock \doi{10.1017/CBO9780511811654}.

\bibitem[Solan and Vieille(2015)]{Solan13743}
Eilon Solan and Nicolas Vieille.
\newblock Stochastic games.
\newblock \emph{Proceedings of the National Academy of Sciences}, 112\penalty0
  (45):\penalty0 13743--13746, 2015.
\newblock ISSN 0027-8424.
\newblock \doi{10.1073/pnas.1513508112}.

\bibitem[Syrgkanis et~al.(2015)Syrgkanis, Agarwal, Luo, and
  Schapire]{SyrgkanisFast}
Vasilis Syrgkanis, Alekh Agarwal, Haipeng Luo, and Robert~E. Schapire.
\newblock Fast convergence of regularized learning in games.
\newblock \emph{CoRR}, abs/1507.00407, 2015.

\bibitem[von Stengel and Forges(2008)]{stengel}
Bernhard von Stengel and Françoise Forges.
\newblock Extensive-form correlated equilibrium: Definition and computational
  complexity.
\newblock \emph{Mathematics of Operations Research}, 33, 11 2008.
\newblock \doi{10.1287/moor.1080.0340}.

\bibitem[{Zhang} et~al.(2018){Zhang}, {Yang}, and {Basar}]{ZYB18}
K.~{Zhang}, Z.~{Yang}, and T.~{Basar}.
\newblock Networked multi-agent reinforcement learning in continuous spaces.
\newblock In \emph{2018 IEEE Conference on Decision and Control (CDC)}, pages
  2771--2776, 2018.

\bibitem[Zhang et~al.(2018)Zhang, Yang, Liu, Zhang, and Basar]{ZYL+18}
Kaiqing Zhang, Zhuoran Yang, Han Liu, Tong Zhang, and Tamer Basar.
\newblock Finite-sample analyses for fully decentralized multi-agent
  reinforcement learning.
\newblock \emph{CoRR}, abs/1812.02783, 2018.

\bibitem[Zhang et~al.(2019)Zhang, Yang, and Başar]{zhang2019multiagent}
Kaiqing Zhang, Zhuoran Yang, and Tamer Başar.
\newblock Multi-agent reinforcement learning: A selective overview of theories
  and algorithms, 2019.

\end{thebibliography}
\clearpage
\appendix

\section{Omitted Proofs}
\label{sec:appendix}

In \Cref{subsec:proofs-adv-mdp}, we show hardness for the ``adversarial MDP'' problem. In \Cref{sec:proofs-gwsr}, we analyze the use of bandit algorithms for reaching correlated equilibria in games with stochastic rewards. In \Cref{sec:bayes}, we introduce a general formulation of Bayesian games, for which obtain analogues of the results in \Cref{sec:proofs-gwsr}, which will be later used for analysis of learning in ``single-controller'' stochastic games. We prove our main results regarding PLL in \Cref{sec:proofs-fhsg}, and FastPLL in \Cref{sec:proofs-fast-fhsg}. Our single-controller result is shown in \Cref{sec:proofs-sc-fhsg}, and our ``shared randomness'' result for extending PLL is given in \Cref{subsec:proofs-sr}.

\subsection{Proofs for Section 3: Hardness of Learning in Adversarial MDPs}
\label{subsec:proofs-adv-mdp}
Here we prove our hardness result for the finite-horizon or ``episodic'' adversarial MDP problem, where both transitions and rewards can change arbitrarily between episodes. We assume the adversary can pick the starting state as well, which is without loss of generality up to increasing the horizon by 1. This problem was shown to be at least as hard as learning parities with noise by Abbasi-Yadkori et al.~[2013], and their reduction involves creating episodic MDPs with $H= \Theta(S)$. We strengthen this to $\textsf{NP}$-hardness, and for a horizon length of only 3. We do this by showing that the batch version of the problem, which we call the ``multi-MDP'', problem is at least as hard as $\textsf{3-SAT}$, and as such is $\textsf{NP}$-hard to approximate within a factor of $\frac{7}{8} + \epsilon$, for any $\epsilon > 0$. By an online-to-batch reduction, this implies that there is no algorithm for the episodic adversarial MDP problem with poly-time per-round computation and 
$O(T^{1-\delta} \poly(S))$ regret, for any $\delta > 0$, and for any dependence on $N$ and $H$,
unless $\textsf{NP} \subseteq \textsf{BPP}$. Like Abbasi-Yadkori et al.~[2013], our reduction only needs deterministic transitions, and so the hardness result also holds for the simpler ``adversarial online shortest path problem''.

\subsubsection{The Offline Problem.}

Consider the batch version of the adversarial MDP problem, which we call the ``multi-MDP problem'', where we are given a set of MDPs $\M$. Each MDP $M \in \M$ has a identical state and action spaces $\X$ and $\A$, as well as episode length $H$, but the transition and reward functions $p$ and $r$ can differ arbitrarily. Given $\M$ as input, the goal for the maximization problem is to output a single (possibly randomized and non-stationary) policy $\pi$ which maximizes the average per-episode reward across all MDPs. The decision problem is to determine if any single policy achieves average reward at least $R$ across $\M$. We assume that the per-episode reward in each MDP $M$ is in $[0,d]$ for all policies, and that all instantaneous rewards are non-negative. 

\begin{theorem}\label{thm:multi-mdp}
The decision version of the multi-MDP problem is $\textsf{\textup{NP}}$-complete for horizon length $H \geq 3$. Further, the maximization version is $\textsf{\textup{NP}}$-hard to approximate within a factor of $\frac{7}{8} + \epsilon$, for all $\epsilon > 0$.
\end{theorem}
\begin{proof}
We reduce from $\textsf{3-SAT}$. First we constrain ourselves to only considering deterministic policies. 
The idea is to encode each of the $m$ clauses of a $\textsf{3-SAT}$ formula (on $n$ variables) as set of six $(n+1)$-state MDPs. 
The states correspond to each of the variables as well as a ``done'' state, and the action space at each state is $\{0,1\}$, corresponding to an assignment for the variable.
Assume without loss of generality that the variables in the input formula are lexicographically ordered. Create one MDP for each of the six possible permutations of the literals in a clause; the episode will consist of three steps. For each of these MDPs, let the starting state $s_i$ at step $h=1$ correspond to the first literal $x_i$ in the ordering. If $x_i$ evaluates to True on input $\pi(s_i, 1)$ for a policy $\pi$, we transition to the ``done'' state, otherwise we transition to the state for the second literal $s_j$. Transitions proceed here accordingly for $\pi(s_j, 2)$ and likewise at the third state $s_k$ for $\pi(s_k, 3)$. Once at the ``done'' state, we remain there until the end of the episode regardless of action. Transitioning to the ``done'' state from some other state yields a reward of 1 and all other transitions yield a reward of 0. 

If the input formula is satisfiable, then the stationary policy corresponding to the satisfying assignment will clearly obtain an average reward of 1. 
A non-stationary policy $\pi$ defines six (not necessarily distinct) assignments of values to the $n$ variables, for each permutation of the 3 timesteps. 
We can split the $6n$ MDPs into 6 sets, each of size $n$ corresponding to one permutation, which are evaluated on the appropriate assignment of values. If the input formula is unsatisfiable, at least one MDP in each set will result in a reward of 0
Deciding whether any policy achieves an average reward of 1 or at most $1-\frac{1}{n}$ is clearly in $\textsf{NP}$, as the best policy acts as a certificate, and so the problem is $\textsf{NP}$-complete. 

This reduction also implies hardness of approximation.
As is well-known, it is $\textsf{NP}$-hard to approximate $\textsf{MAX-3-SAT}$ within a factor of $\frac{7}{8} + \epsilon$, for all $\epsilon > 0$.
Suppose we can could compute a policy which obtains average reward at least $\frac{7}{8} +\epsilon$ in a set of MDPs with maximum possible average reward of 1. We can then apply to the above reduction to any input $\textsf{3-SAT}$ formula, resulting in a set of MDPs with a possible average reward of 1 if and only if the formula is satisfiable. If we can obtain average reward at least $\frac{7}{8} +\epsilon$ on this set, we must have average reward at least $\frac{7}{8} + \epsilon$ on the subset of MDPs corresponding to some permutation of literals. We can then extract an assignment from that permutation of timesteps in the policy which corresponds to an assignment which satisfies at least a $\frac{7}{8} + \epsilon$ fraction of the clauses in the input formula, implying the desired hardness result.

Any randomized policy can be derandomized without loss in average reward in polynomial time, implying that randomization does not help from a complexity perspective.
\begin{lemma}
For any set of finite-horizon MDPs, any randomized policy can be converted to a deterministic non-stationary policy in polynomial time without decreasing average reward. 
\end{lemma}
\begin{proof}
Consider the uniform distribution over MDPs in the set $\M$ and the induced distribution over states in the final timestep. By the Markov property and the assumption of a fixed policy, the conditional distribution of actions at a state is independent of the MDP as well as the sequence of states visited. Each action with positive support has some expected reward when taking the expectation over the MDP distribution, transitions, and previous action selections; playing the maximum action at each state does not decrease expected reward. We can apply this to each previous step by backward induction, as downstream conditional expected values for actions at each state are still defined, giving us a fully deterministic policy.
\end{proof}
As such, the hardness result holds even for algorithms which output randomized non-stationary policies.

\end{proof}

\subsubsection{Hardness for Regret Minimization and Black-Box NFCCEs in Stochastic Games}

We use Theorem 9 to prove our hardness result for quickly vanishing regret in the adversarial MDP problem.

\paragraph{Restatement of Theorem 2.}
\emph{
Assuming $\textsf{\textup{NP}} \not\subseteq \textsf{\textup{BPP}}$, there is no algorithm with polynomial time per-round computation which has  $O(T^{1 - \delta}\textup{poly} \cdot (S))$ regret algorithm for the adversarial MDP problem with $H \geq 3$, for any $\delta > 0$.
}

\begin{proof}
By the standard online-to-batch reduction from \cite{CCG04}, we can convert an algorithm with small regret to an algorithm for $\textsf{MAX-3-SAT}$. 
Suppose we had an algorithm with regret $O( T^{1 - \delta} S^k)$ for constants $k$ and $\delta$. Take $T \gg S^{k/\delta}$ but still polynomial in $S$ such that the average regret is $o(1)$. Apply the reduction from Theorem 9 to a $\textsf{3-SAT}$ instance on $S$ variables and then run the algorithm for $T$ steps, sampling from the uniform distribution over the constructed MDPs at each episode. By the main result (Theorem 4) from \cite{CCG04}, the empirically optimal policy over the historical sequence achieves a value within $o(1)$ of the optimum with high probability.
This would imply a polynomial time algorithm which beats a $\frac{7}{8} + \epsilon$ approximation for $\textsf{MAX-3-SAT}$, which is impossible unless $\textsf{\textup{NP}} \subseteq \textsf{\textup{BPP}}$ due to \cite{Has97}.
\end{proof}

This directly implies Corollary 2.1, where the horizon is increased to 4 to account for the starting state in a finite-horizon stochastic game being random rather than adversarial (in our reduction, one can add a ``starting state'' from which the adversary selects the next state).

\subsection{Games with Stochastic Rewards}
\label{sec:proofs-gwsr}

Recall that for a game with stochastic rewards, we consider all players running an adversarial bandit algorithm $\B$ (such as SR-MAB). A step in our analysis introduces an additional $\log(1/\epsilon)$ term beyond the runtime of SR-MAB for target average regret $\epsilon$, yet with less dependence on $N$. This is not an issue if $N$ is sufficiently large as a function of $\epsilon$, but if this is not the case we extend the runtime to that which would be required if $N = \Omega\left(\sqrt[3]{ \frac{\log(1/\epsilon)}{\log({\log(1/\epsilon)})} } \right)$, which can only increase average regret;
we denote this runtime function by $B(\epsilon, N)$.

\paragraph{Theorem 3.}
\emph{When players in a game with stochastic rewards $x$ select actions using $\B$ for $T \geq B(\epsilon/4, N)$ rounds, the sequence of action profiles is an $\epsilon$-correlated equilibrium for the game, where the expectation is taken with respect to the tensor distribution as well as $\B$. }

\begin{proof}[Proof of Theorem 3]
We begin with a lemma relating the runtime of SR-MAB to the term which we will use in our martingale analysis of the ``sampling error'' of the realized sequence of reward tensors versus the average tensor $\bar{\theta}$.

\begin{lemma}\label{lemma:action-dupl}
If
$N = \Omega\left(\sqrt[3]{ \frac{\log(1/\epsilon)}{\log({\log(1/\epsilon)})} } \right)$
then  $\frac{N^3 \log(N)  }{\epsilon^2} 
= \Omega\parens{\frac{N \log(N) + \log(1/\epsilon)}{\epsilon^2}}$.
\end{lemma}

\begin{proof}[Proof of \Cref{lemma:action-dupl}]
It suffices to show that 
$N^3\log(N) = \Omega(\log(1/\epsilon))$.
Plugging in our expression for $N$, we have that
\begin{align*}
    N^3 \log N =&\; \Theta(N^3 \log(N^3) ) \\
    =&\; \Omega\left(\frac{\log(\frac{1}{\epsilon}) \cdot \left( \log\log(\frac{1}{\epsilon})- \log \log \log(\frac{1}{\epsilon}) \right) }{\log \log (\frac{1}{\epsilon})} \right) \\
    =&\; \Omega(\log(1/\epsilon)).
\end{align*}
\end{proof}
By the regret guarantee of $\B$, each player has expected average swap regret at most $\epsilon/4$ with respect to the sampled sequence of reward tensors $\{\theta^t\}_{t \in [T]}$, which we denote $\BReg_{\B}^{\{\theta^t\}}$.
For a player $i$, consider some swap function $f$. Let $X_{f}^t = u_i(f(a^t_i); a^t_{-i}, \theta^t ) - u_i(f(a^t_i); a^t_i, \bar{\theta})$ for an action profile and tensor $(a^t, \theta^t)$, i.e.\ the difference between this player's reward from using $f$ on $\theta^t$ versus the average tensor $\bar{\theta}$, given the action profile $a^t$. Let $Y_f^t = \sum_{j=1}^t X^t_f$. 
For a distribution over tensors, and any sequence of action profiles where $a_t$ is independent of $\theta^t$ given actions and tensors for $1,\ldots,t-1$, 
the sequence $Y_f^1,\ldots,Y_f^t$ is a martingale with respect to the sequence $X^f_t$. 
To see this, note that for any fixed $a^t$, 
$X_{f}^t$ 
is in $[-1,1]$ as rewards are in $[0,1]$,
and $\E[Y_{f}^t \mid  X_{f}^1,\ldots,X^{t-1}_{f}] = Y_{f}^{t-1}$, as
$\E[X_{f}^t \mid X_{f}^1,\ldots,X^{t-1}_{f}] = 0$
by the definition of $\bar{\theta}$. 

Let $T \geq \frac{32(N\log(N) + \log(8/\epsilon)) }{\epsilon^2} = \frac{32 \log(8N^N / \epsilon) }{\epsilon^2}$ by \Cref{lemma:action-dupl}.
By the Azuma-Hoeffding inequality we have that
\begin{align*}
    \Pr[\abs{Y^T_f} \geq \frac{\epsilon T}{4}] \leq&\; 2 \exp\left(\frac{-\epsilon^2 T}{32} \right) \\
    \leq&\; \frac{\epsilon}{4 N^N}.
\end{align*}

Union-bounding over all $f \in \Fsw$, we then have that 
\begin{align*}
    \max_{f \in \Fsw} ~ \abs{ \frac{1}{T}\sum_{t=1}^T u_i(f(a^t_i); a^t_{-i}, \theta^t ) - u_i(f(a^t_i); a^t_i, \bar{\theta})} \leq&\; \frac{\epsilon}{4} 
\end{align*}
with probability at least $1 - \frac{\epsilon}{4}$. As such, the average utility of a swap function on the sequence deviates from its expected utility on the distribution by at most $\epsilon/4$ with probability at least $\epsilon/4$, holding simultaneously for all functions, including the identity function $I$ (our benchmark for swap regret). 
As such, with probability $1-\epsilon/4$, 
the difference in swap regret on the sequence and the distribution, denoted by $\abs{ \BReg_{\B}^{\bar{\theta}} - \BReg_{\B}^{\{\theta_t\}} } $, is at most $\epsilon/2$.
Using the maximal deviation of 1 as a bound for the difference for the remaining probability, we then have that  
\begin{align*}
    \EE{ \abs{ \BReg_{\B}^{\bar{\theta}} - \Reg_{\B}^{\{\theta^t\}} } } \leq&\; (1 - \epsilon/4)\cdot \epsilon/2 + \epsilon/4 \leq \frac{3\epsilon}{4}. 
\end{align*}
Therefore by our bound on $\EE{\BReg_{\B}^{\{\theta^t\}}}$ and linearity of expectation:
\begin{align*}
    \EE{\BReg_{\B}^{\bar{\theta}} } =&\; \EE{\BReg_{\B}^{\{\theta^t\}} }  + \EE{  \BReg_{\B}^{\bar{\theta}} - \BReg_{\B}^{\{\theta^t\}} }  \\
    \leq&\; \EE{\BReg_{\B}^{\{\theta^t\}} }  + \EE{ \abs{ \BReg_{\B}^{\bar{\theta}} - \BReg_{\B}^{\{\theta^t\}} } } \\
    \leq&\; \epsilon. 
\end{align*}
As no player can improve average utility in expectation for $\bar{\theta}$ by more than $\epsilon$ with any swap function, the uniform distribution over the sequence of action profiles is an $\epsilon$-correlated equilibrium for $x$ when taking the expectation over both the profile sequence and the generating process using $\B$ and samples of reward tensors.  
\end{proof}

\begin{corollary}[Restatement of Corollary 3.1]
 When all agents in a game with stochastic rewards $x$ play according to $\B$ for at least $\frac{2 \log\parens{5M/\delta} }{\eta^2} \cdot B(\epsilon/8, N)$ rounds, simultaneously restarting $\B$ every $B(\epsilon/8, N)$ rounds, the resulting sequence of actions is an $(\epsilon/2 + \eta/2)$-correlated equilibrium for $x$ with probability at least $1 - \delta/5$.
 
 Further, let $V_i^{\B}(x) = \E_{\B,r}\left[\frac{1}{T}\sum_{t=1}^T u_i(a_i^t; a^t_{-i}, \theta^t) \right]$ and let $\hat{V}_i^{\B}(x)$ be the average utility received by player $i$ over all rounds.  
 With probability at least $1 - 2 \delta/5$, $\abs{V_i^{\B}(x) - \hat{V}_i^{\B}(x) } \leq \eta/2$ simultaneously for all players.

 Additionally, the computed estimate is within $\eta$ of player $i$'s expected average reward for playing the game according to the resulting policy distribution with probability at least $1 - 2\delta/5$. 

\end{corollary}
\begin{proof}[Proof of Corollary 3.1]
    The swap regret of a sequence is upper-bounded by the sum of the swap regret values of a uniform partition of the sequence, as the latter may use a different swap function on each sequence while the former is restricted to only using a single function. As such, we can bound the average regret of our sequence by averaging the average swap regret values between restarts.  
    
    Both average utility and average swap regret (with respect to $x$) over $B(\epsilon/8)$ are random variables taking values in $[0,1]$, and the mean of the latter is at most $\epsilon/2$ by Theorem 3. 
    Recall from the proof of Theorem 3 that the expected average reward deviation of the identity function on the sequence and distribution 
    has mean zero (by nature of it being a martingale), and it takes values in $[-1,1]$.
    The result then follows from applying Hoeffding's inequality to the average of the samples we receive of the random variables, bounding deviation by $\eta/2$ (or $\eta$), and union-bounding over all players and failure probabilities. 

\end{proof}

\subsection{Correlated Equilibria in Bayesian Games}
\label{sec:bayes}

We also give a convergence result for 
learning in Bayesian games. The Bayesian game formulation we consider is quite general (in particular, we remove the ``independent private value'' assumption from the model considered in \cite{HST15}, and allow signals and rewards to be arbitrarily correlated across players), and can be viewed as a partial-information generalization of games with stochastic rewards.
When all players use our described method, the sequence of policy profiles played by all players converges to an approximate Bayes correlated equilibrium in polynomial time. 

\begin{definition}[Bayesian Games]
    A Bayesian game $y = (\A, M, \psi, p, r, u)$ has $M$ players and is specified by a set of action profiles $\A = \times_{i \in [M]} \A_i$, a signal function $\psi : \X \rightarrow \Psi$ where $\Psi = \times_{i \in [M]} \Psi_i$, and a distribution over states $p \in \Delta(\X)$. Each state $x$ denotes a game with stochastic rewards, with its distribution over reward tensors given by $r: \X \rightarrow \Delta(\Theta)$. 
   Players' utilities, given by $u : \A \times \Theta \rightarrow [0,1]^{N}$, depend on the realization of $\theta \sim r(x)$. Players only observe a signal of the state $\psi_i(x)$, and never observe $\theta$ or $x$ directly.
\end{definition}
We assume that $\abs{\A_i}= N $ for all agents, and we will let $S_i = \abs{\Psi_i}$ and $S = \max_i S_i$.
In this model of a Bayesian game, a state $x$ is drawn from $p$, each agent $i$ observes a signal $\psi_i(x)$ and selects an action $a_i$, then receives utility $u_i(a_i; a_{-i}, \theta)$, where $\theta$ is drawn from $r(x)$.  
We note that Bayesian games are often defined in such a way where states and reward tensors are treated as equivalent. 
This formulation of a Bayesian game is similar to the ``information set'' model often considered in partially-observable Markov decision processes and extensive-form games.
However, our result for Bayesian games will not depend on the size of $\X$ or $\Theta$. 
Here, one could treat $\X$ and $\Theta$ as identical, but we maintain the distinction for continuity in exposition with our sections on stochastic games.
It is without loss of generality that we assume $u$ depends only on $a$ and $\theta$, not $x$, as we can encode arbitary distributions over reward vectors in $[0,1]^M$ for each state with a distribution over reward tensors.

The definition of correlated equilibrium in Bayesian games given in \cite{doi:10.3982/TE1808} refers to a {\it decision rule}, given by a distribution over action profile recommendations for each state and set of types, which is {\it obedient} in the sense that no player can improve by deviating from the recommendations for any action-type pair.
The method we present here will converge to a joint distribution over policy profiles, denoting an action recommendation for each signal, which will be independent of the state and reward tensor distributions, and which satisfies this definition of Bayes correlated equilibrium. 
Several other definitions are considered in the literature as well \cite{RePEc:cor:louvco:1993009}.

We are aware of only one paper, \cite{HST15}, which considers learning correlated equilibria in Bayesian games through the lens of polynomial time convergence, where the primary focus is on analyzing the Price of Anarchy and connections to learning in auctions.
They consider the {\it independent private value model} of Bayesian games, 
There, the assumption is made that players have ``types'' which fully characterize their rewards for any action profile, and further that these types are drawn from a product distribution. In their approach, each agent runs parallel copies of a no-regret algorithm for each type, and actions are sampled from each algorithm every round, which they interpret as the sampling of a strategy mapping types to actions.   
Our model is a generalization of this setting, as we allow types (signals) to be arbitrarily correlated with eachother as well as with the reward tensors.
To our knowledge, the approach we give here is the first which converges to a Bayes correlated equilibrium in polynomial time for such a general formulation of Bayesian games.

Here will consider {\it policies} $\pi_i : \Psi_i \rightarrow \A_i$ for an agent $i$, with $\pi_i \in \Pi_i$ and $\Pi = \times_{i \in [M]} \Pi$,  which are functions mapping their signals to actions. In our setting, a Bayes correlated equilibrium is a distribution over policy profiles such that no agent can benefit by deviating from policy recommendations.

\begin{definition}[Bayes Correlated Equilibria]
    A Bayes correlated equilibrium for a Bayesian game is a distribution over policy profiles given by $D \in \Delta(\Pi)$ such that for all players $i$ and all swap functions $f \in \Fsw^{\Psi_i} : \A_i \times \Psi_i \rightarrow \A_i$,
\begin{align*}
    \E_{\pi \sim D, \theta \sim r(x), x \sim p}[U_i] \geq \E_{a \sim D(x), \theta \sim r(x), x \sim p}[U_i^f], 
\end{align*}
with $U_i = u_i(\pi_i(\psi_i(x)); a_{-i}, \theta)$ and $U_i^f = u_i(f(\pi_i(\psi_i(x)), \psi_i(x)); a_{-i}, \theta)$,
where $a_{-i}$ is the vector of actions $[\pi_j(\psi_j(x))]$ for agents $j \neq i$,
and where the policy vector $\pi$ is sampled independently from $x$.
Such a distribution is an $\epsilon$-Bayes correlated equilibrium if for all players and swap functions,
\begin{align*}
    \E_{\pi \sim D, \theta \sim r(x), x \sim p}[U_i] \geq \E_{a \sim D(x), \theta \sim r(x), x \sim p}[U_i^f] - \epsilon. 
\end{align*}
The smallest quantity $\epsilon$ for which the above holds for agent $i$ is their average $\Fsw^{\Psi_i}$-regret for a policy distribution.
\end{definition}

We let $\B_S$ denote the {\it parallel bandit} algorithm consisting of $S$ copies of $\B$, with one copy for each type. At the beginning of each round, agents sample actions from each copy of $\B$, thereby creating a policy $\pi_i^t$ for the round. Upon observing their signal $\psi^t_i$, they play the action $\pi_i^t(\psi)i^t)$, update the copy of $\B$ corresponding to $\psi_i$ with their observed reward, and record a record a reward of 0 for all other copies. 
We show that when agents play according to $\B_S$, the sequence of policies converges to an approximate equilibrium for the Bayesian game. 

\begin{theorem}
    \label{thm:bayesian-games-ce}

    When players in a Bayesian game $y$ select actions using $\B_S$ for $T \geq B(\frac{\epsilon}{4S}, N)$ rounds, where the state is sampled independently each round and the reward tensor is sampled from that state's distribution, the sequence of policies is an $\epsilon$-correlated equilibrium for the game,  
    where the expectation is taken with respect to the state, tensor, and action profile distribution as well as the randomness of $\B$. 

\end{theorem}
\begin{proof}[Proof of \Cref{thm:bayesian-games-ce}]
    The proof is quite similar to that for Theorem 3. We bound the expected average swap regret for each copy of $B$ by $\epsilon/S$, which then bounds the total average swap regret (with respect to the policy class) by $\epsilon$.

    By the guarantee of the algorithm $\B$, each player's copy of $\B$ for a signal $\psi_i$ has expected swap regret at most $\frac{\epsilon}{4S}$ with respect to the sampled sequence of states and reward tensors (where rewards are 0 when the corresponding signal is not observed), which we denote $\BReg_{\B, \psi_i}^{\{\theta^t, x^t\}}$. The average swap regret for the entire sequence will be the sum of the swap regrets for each signal, denoted $\BReg_{\B_S}^{\{\theta^t, x^t\}} = \sum_{\psi_i \in \Psi_i} \BReg_{\B, \psi_i}^{\{\theta^t, x^t\}}$, as the deviations considered by the function class $\Fsw^{\Psi_i}$ are equivalent to choosing any $f \in \Fsw$ for each signal.

For a player $i$ and signal $\psi_i$, 
upon fixing the vector of opponent policies $\pi_{-i}$, there is some fixed expected reward for each action, conditional on observing $\psi_i$, given by:
\begin{align*}
    \bar{\theta}(a_i ; \psi_i, \pi_{-i}) =&\; \E_{x\sim p,\theta\sim r(x)}\left[U_i^{\psi} \times \mathbf{1}[\psi_i(x) = \psi_i] \right],
\end{align*}
where $U_i^{\psi} = u_i(a_i; (\pi_j(\psi_j(x)))_{j \neq i}, \theta)$.
In round $t$ of the game, the reward that player $i$'s copy of $\B$ associated with $\psi_i$ will receive for playing action $a_i$ is a random variable in $[0,1]$ with mean $\bar{\theta}(a_i ; \psi_i, \pi_{-i})$,  
where we view $\pi_{-i}$ as being fixed prior to the realization of $x$ and $\theta$. 
The regret bound for that copy of $\B$ holds for the realized sequence of vectors (determined by $\pi_{-i}^t$, $x^t$, and $\theta^t$) of these rewards for all actions $a_i$.
We will be interested in bounding the average reward deviation of swap functions between this sequence and the sequence $((\bar{\theta}(a_i ; \psi_i, \pi_{-i}^t))_{a_i \in \A_i})_{t \in [T]}$.

Consider some swap function $f \in \Fsw$. 
We can again define a martingale which tracks the deviation of the performance of $f$ on the sampled sequence versus the underlying game distribution.  
Let $X_{f, \psi_i}^t = (u_i(f(a^t_i); (\pi^t_j(\psi_j(x^t)))_{j\neq i}, \theta^t )\cdot \mathbf{1}[\psi_i(x^t) = \psi_i] - \bar{\theta}(f(a^t_i); \psi_i, \pi^t_{-i})) $
for a policy profile, signal, and tensor $(\pi_{-i}^t, \psi_i^t, \theta^t)$, i.e.\ the difference between this player's observed and expected reward from using $f$ with the copy of $\B$ associated with $\psi_i$, given opponent policies $\pi_{-i}^t$ and their own sampled action $a_i^t$ for signal $\psi_i$.
Let $Y_{f,\psi_i}^t = \sum_{j=1}^t X^t_{f, \psi_i}$. 
For a distribution over states and tensors, and any sequence of action profiles where $a_t$ is independent of $\theta^t$ given actions and tensors for $1,...,t-1$, 
the sequence $Y_{f, \psi_i}^1,...,Y_{\psi_i}^t$ is a martingale with respect to the sequence $X_{f, \psi_i}^t$. 
To see this, note that for any fixed $a^t$, 
$X_{f, \psi_i}^t$ 
is in $[-1,1]$ as rewards are in $[0,1]$,
and $\E[Y_{f,\psi_i}^t \mid X_{f, \psi_i}^1,\ldots,X^{t-1}_{f, \psi_i}] = Y_{f,\psi_i}^{t-1}$, as
$\E[X_{f,\psi_i}^t \mid X_{f, \psi_i}^1,\ldots,X^{t-1}_{f, \psi_i}] = 0$
by the definition of $\bar{\theta}$. 

Let $T \geq \frac{32S^2 (N\log(N) + \log(8S /\epsilon)) }{\epsilon^2} = \frac{32 S^2 \log(8 S N^N / \epsilon) }{\epsilon^2}$ by \Cref{lemma:action-dupl}.
By the Azuma-Hoeffding inequality we have that
\begin{align*}
    \Pr[\abs{Y^T_{f, \psi_i}} \geq \frac{\epsilon T}{4S}] \leq&\; 2 \exp\left(\frac{-\epsilon^2 T}{32 S^2 } \right) \\
    \leq&\; \frac{\epsilon}{4 S N^N}.
\end{align*}
Union-bounding over all $f \in \Fsw$, we then have that 
\begin{align*}
    \max_{f \in \Fsw} ~ \abs{ \sum_{t=1}^T 
    \frac{U_i^{\psi,t} \cdot \mathbf{1}[\psi_i(x^t) = \psi_i] - \bar{\theta}(f(a^t_i); \psi_i, \pi^t_{-i})}{T} }   
    \leq&\; \frac{\epsilon}{4S} 
\end{align*}
where $U_i^{\psi,t} = u_i(f(a^t_i); (\pi^t_j(\psi_j(x^t)))_{j\neq i}, \theta^t )$
with probability at least $1 - \frac{\epsilon}{4S}$. As such, the average utility of a swap function on the sequence applied to the copy of $\B$ for $\psi_i$ deviates from its expected utility on the distribution by at most $\frac{\epsilon}{4S}$ with probability at least $\frac{\epsilon}{4S}$, holding simultaneously for all functions in $\Fsw$, including the identity function $I$ (our benchmark for swap regret). 
As such, with probability $1-\frac{\epsilon}{4S}$, 
the difference in average swap regret on the sequence and the distribution for this $\B$ copy, denoted by $\abs{ \BReg_{\B, \psi_i}^{\bar{\theta}} - \BReg_{\B, \psi_i}^{\{\theta^t, x^t\}} } $, is at most $\frac{\epsilon}{2S}$.
Using the maximal deviation of 1 as a bound for the difference for the remaining probability, we then have that  
\begin{align*}
    \EE{ \abs{ \BReg_{\B, \psi_i}^{\bar{\theta}} - \Reg_{\B, \psi_i}^{\{\theta^t, x^t\}} } } \leq&\; (1 - \frac{\epsilon}{4S})\cdot \frac{\epsilon}{2S} + \frac{\epsilon}{4S} \leq \frac{3\epsilon}{4}. 
\end{align*}
Therefore by our bound on $\EE{\BReg_{\B}^{\{\theta^t, x^t\}}}$ and linearity of expectation:
\begin{align*}
    \EE{\BReg_{\B, \psi_i}^{\bar{\theta}} } =&\; \EE{\BReg_{\B, \psi_i}^{\{\theta^t,x^t \}} }  + \EE{  \BReg_{\B, \psi_i}^{\bar{\theta}} - \BReg_{\B, \psi_i}^{\{\theta^t, x^t\}} }  \\
    \leq&\; \EE{\BReg_{\B, \psi_i}^{\{\theta^t, x^t\}} }  + \EE{ \abs{ \BReg_{\B, \psi_i}^{\bar{\theta}} - \BReg_{\B, \psi_i}^{\{\theta^t, x^t\}} } } \\
    \leq&\; \epsilon. 
\end{align*}
Summing over each copy of $\B$ gives us that $\EE{\BReg_{\B_S}^{\{\theta^t, x^t\}}} \leq \epsilon$, as average swap regret (with respect to $\Fsw^{\Psi_i}$) for the distribution can be decomposed into swap regret for each signal (with respect to $\Fsw$) just as for the sequence of states and tensors. 
As no player can improve average utility in expectation for $\bar{\theta}$ by more than $\epsilon$ with any swap function $\Fsw^{\Psi_i}$, the uniform distribution over the sequence of policy profiles is an $\epsilon$-correlated equilibrium for $y$ when taking the expectation over both the profile sequence and the generating process using $\B_S$ and samples of states and reward tensors.  
\end{proof}

Again, if desired we can simultaneously obtain an accurate estimate of the value 
$V_i^{\B_S}(y)$ 
of this equilibrium-generating process for each player, and boost regret bounds to high probability, with repeated restarts.

\subsection{Analysis for PLL}
\label{sec:proofs-fhsg}

Showing \Cref{thm:bill} for BILL is straightforward and a proof can be obtained by simplifying the analysis of PLL in \Cref{thm:fhsg-ce}. We restate the description of PLL here, with explicit constants for the terms whose asymptotic descriptions were given in the body.

\paragraph{Algorithm 2: Parallel Local Learning.}
Initialize $\hat{V}_i^{\B}(x, h) = H-h+1$ for each pair $(x, h)$, as well as a visit counter $c(x,h)$ for each pair set to 0. 
Let $W = \max\parens{W_1, W_2}$, where 
$W_1 = \frac{128 S^4 H^6 \log(2S / \delta') }{\epsilon^2}$
, $W_2 = \frac{512 H^4 \log(5M  / \delta' ) }{\epsilon^2}$, and $\delta' = \frac{\epsilon \delta}{192 SH^4 \parens{(S+1)^H + 1} \cdot \max\parens{S , 4 H^7/{\epsilon}}}$.
Let $L \geq \max\parens{\frac{64 S^2 H^3 W B}{\epsilon} , \frac{256 SH^4 W B}{\epsilon^2}}$.
Initialize a copy of $\B$ at each pair, specified to run for $B = B(\frac{\epsilon}{16H}, N)$ steps.
 Until termination, run the following procedure for each epoch:
\begin{itemize}
\item Run for $L$ trajectories, using $\B$ at each pair, counting rounds and updating actions for a copy of $\B$ only when the corresponding pair is visited. Record rewards as the sum of the observed reward as well as the value estimate for the next pair visited in that trajectory, scaled to [0,1]. 
\item Consider the last step  $h \in [H]$ where an unlocked pair's counter crossed $\frac{16H^2WB}{\epsilon}$ in the epoch. 
Lock all unlocked states at this step with appropriate estimates which were previously unlocked, compute value estimates $\hat{V}_i^{\B}(x, h)$ as the average reward over the corresponding
$\frac{16H^2WB}{\epsilon}$ visits,
then reset all copies of $\B$, counters, and estimates at {\it earlier} pairs $(h' < h)$.
\item Terminate if no pair's counter crosses $\frac{16H^2WB}{\epsilon}$ in the epoch.
\end{itemize}

\paragraph{Restatement of Theorem 5.}\emph{
    PLL terminates after at most $(S+1)^H+1$ epochs.
    After termination, for each pair (x, h), consider the uniform distribution over action profiles $D(x,h)$ played since that pair was last reset. 
    Let $D$ be the distribution over policy profiles where the action profile for each pair $(x,h)$ is sampled independently from $D(x, h)$. 
With probability at least $1 - \delta$, $D$ is an $\epsilon$-EFCE for the game.
}

\begin{proof}[Proof of Theorem 5]

We first give a worst-case bound on the runtime, then proceed with our analysis of the regret of the resulting action profile distributions. At termination, for any pair $(x,h)$ with no visits since it was last reset, we can let the distribution $D(x,h)$ over action profiles be arbitrary for the purposes of our our analysis.

\paragraph{\Cref{lemma:pll-epochs}} \emph{PLL runs for at least $H$ epochs, and at most $(S+1)^H+1$ epochs. }
\begin{proof}
    All pairs start unlocked, and some pair in each step is visited at least $\frac{16H^2 WB}{\epsilon}$ per epoch by pigeonhole, so the algorithm will not terminate unless there is a locked pair for every step. 
States are only moved from unlocked to locked at one step per epoch, and so there must be at most $H$ epochs to lock some pair in all steps.

We can bound the number of epochs by bounding the number of epochs in which a pair at some step can become locked. Observe that a locked pair at step $H$ will only become locked in one epoch and will never become unlocked afterwards. A pair at step $H-1$ will become locked in at most $S$ epochs, as it will only become locked after at least one pair at step $H$ is locked, and then can be unlocked at most $S-1$ times for the remaining unlocked pairs at step $H$.   
In general, the number of epochs in which a state can become locked is bounded by the number of epochs in which a downstream state can become locked. 
Let $g(h)$ denote this bound on the number of epochs in which a pair at step $h$ can be locked, which is given by:
\begin{align*}
    g(h) =&\; \sum_{i = h+1}^H S g(i) \\ 
    =&\; Sg(h+1) + \sum_{i = h+2}^H S g(i) \\ 
    =&\; (S+1)g(h+1) \\ 
    =&\; (S+1)^{H-h} g(H) \\
    =&\; (S+1)^{H-h},
\end{align*}
as $g(H) = 1$.
The total number of epochs before termination is then bounded by
\begin{align*}
    1 + \sum_{i=1}^H Sg(i) = g(0) + 1 = (S+1)^H + 1, 
\end{align*}
accounting for the last epoch in which no states are locked. 
\end{proof}
 
We now show that each agent has small regret with respect to $\Fsw$ under the resulting policy distribution $D$ with high probability, which coincides with the definition of extensive-form correlated equilibria we consider, as $D$ is a product distribution across pairs.
An important object in this analysis is the expected distribution over state visitations when players use $\B$ at each pair with a fixed set of values. Just as there is some fixed distribution over average rewards when players play $\B$ in a game for many rounds, there is also a fixed distribution over transitions when using $\B$ at a pair in a stochastic game, given fixed sets of value estimates for downstream states.

When all agents use a bandit algorithm $\B$ at a pair $(x, h-1)$ for $B$ trajectories {\it where $(x,h-1)$ is visited}, augmenting rewards with downstream value estimates $\hat{V}_i(x', h)$ for each player $i$ and state $x'$, there is some expected proportion of those trajectories that each state will be visited at step $h$, which we denote by:
\begin{align*}
    p_{\hat{V}}(x'; x, h) =&\; \E_{\B, p(h)} \left[ \frac{1}{B} \sum_{t=1}^B \mathbf{1} \left[ \tau(a^t, x) = x' \right] \right]
\end{align*}
We can also define the probability that a pair is visited in a trajectory, assuming that the distribution of transitions between pairs is given by $p_{\hat{V}}$, which we denote by $q_{\hat{V}}(\cdot,\cdot)$:
\begin{align*}
    q_{\hat{V}}(x, 1) =&\; p_0(x), \\
    q_{\hat{V}}(x, h) =&\; \sum_{x' \in \X} q_{\hat{V}}(x', h-1) \cdot p_{\hat{V}}(x; x', h-1).
\end{align*}
For a distribution $D(x, h)$ of action profiles for each pair, we can also define transition probabilities between pairs in a trajectory when action profiles are selected independently for each pair:
\begin{align*}
    p_{D}(x'; x, h) =&\; \Pr_{D(x, h-1), p(h)} \left[   \tau(a, x) = x' \right],
\end{align*}
as well as expected visitation frequencies for each pair in a trajectory:
\begin{align*}
    q_{D}(x, 1) =&\; p_0(x), \\
    q_{D}(x, h) =&\; \sum_{x' \in \X} q_D(x', h-1) \cdot p_{D}(x; x', h-1).
\end{align*}

If a pair $(x,h)$ is visited sufficiently often with fixed downstream values $\hat{V}$, then both the empirical transition distribution and the transition distribution {\it when transition functions are resampled} are close to $p_{\hat{V}}(x,h)$.

We prove a lemma about the composition of bounds on the total variation distance in this setting.

\begin{lemma}
    \label{lemma:tvd-composition}
For distribution functions $p$ and $\hat{p}$ mapping $\X \times [H]$ to $\Delta(\X)$, and $q$ and $\hat{q}$ mapping $[H]$ to $\Delta(X)$, where $q(x, h+1) = \sum_{x' \in \X} q(x', h) \cdot p(x; x', h)$ and $q(x, 1)$ can be arbitrary (and with $\hat{q}$ defined likewise with respect to $\hat{p}$), then
with $d_q^{h+1} = d_{TV}(q(\cdot , h+1 ),  \hat{q}(\cdot , h+1 ))$,
\begin{align*}
    d_q^{h+1} \leq&\; d_{TV}(q(\cdot , h+1), \hat{q}(\cdot , h+1))  \\
    &\;+ \sum_{x' \in \X} q(x', h) \cdot d_{TV}(p(\cdot; x', h), \hat{p}(\cdot; x' , h)).
\end{align*}
\end{lemma}

\begin{proof}[Proof of \Cref{lemma:tvd-composition}]
\begin{align*}
    d_q^{h+1} =&\; \frac{1}{2} \sum_{x \in \X} \abs{ \sum_{x' \in X} q(x' , h ) p(x; x', h)  - \hat{q}(x' , h ) \hat{p}(x; x', h) } \\ 
     \leq&\; \frac{1}{2}   \sum_{x' \in X} \parens{ q(x' , h )\sum_{x \in \X} \abs{ p(x; x', h) - \hat{p}(x; x', h)} }\\
     &\; + \sum_{x' \in X} \parens{ \abs{q(x' , h ) - \hat{q}(x' , h ) } \sum_{x \in \X} p_{D}(x; x', h) } \\
     =&\; d_{TV}(q(\cdot , h), \hat{q}(\cdot , h)) \\
     &\; + \sum_{x' \in X}  q(x' , h ) \cdot d_{TV}(p(\cdot ; x', h), \hat{p}(\cdot ; x', h)).
\end{align*}
\end{proof}

In \Cref{lemma:visit-freq-fhsg} we show that in each epoch, for any pair where $q_{\hat{V}}(x,h)$ is sufficiently large (for the estimates $\hat{V}$ used in that epoch), the number of times in that epoch $(x, h)$ is visited is close to expectation.  
We then show that 
any state which is unlocked at termination will almost surely be visited infrequently when agents play according to $D(x,h)$ at each state.

\begin{lemma}\label{lemma:visit-freq-fhsg}
In any epoch where current value estimates are given by $\hat{V}$ for each player and pair, with probability at least $1 - \frac{\delta}{3\parens{(S+1)^H+1}}$,
every pair $(x,h)$ where $q_{\hat{V}}(x,h) \geq \frac{\eps}{8SH^2}$ reaches the locking threshold by the completion of the epoch.
\end{lemma}
\begin{proof}

We proceed by showing that in each epoch, with high probability, the total variation distance between $q_{\hat{V}}(\cdot, h)$ and the empirical distribution over visited states at step $h$ is small for every $h$.
We prove this inductively.

Consider a sequence of $BW$ visits to a pair $(x, h)$, where $W$ runs of $\B$ are completed. For each run of $\B$, the number of visits to a given pair $(x',h+1)$ is a random variable in $[0,B]$ with mean $B\cdot p_{\hat{V}}(x'; x, h)$, determined by the randomness of each player's copy of $\B$ as well as the game. For such a pair $(x', h+1)$, let $X_i$ denote the $[0,1]$ scaling of this random variable for the $i$th of the $W$ runs, which has mean $p_{\hat{V}}(x'; x, h)$, and let $X = \sum_{i=1}^W X_i$. Each run is independent and so by Hoeffding's inequality,
\begin{align*}
    \Pr\left[\abs{X - \E[X]} \leq \frac{2W \epsilon'}{S} \right] \leq 2\exp{\parens{-8W (\epsilon')^2 / S^2}}
\end{align*}
which is at most $\frac{\delta'}{S}$ if $W \geq \frac{S^2 \log(2S/\delta') }{8 (\epsilon')^2}$.
This holds for all states $x'$ with probability $1-\delta'$ by a union bound, at which point we have that the empirical visitation frequency for every state $x'$ is within $\pm 2\epsilon'/S$ of $p_{\hat{V}}(x'; x,h)$, implying that the total variation distance is at most $\epsilon'$.

Let $\epsilon' = \frac{\epsilon}{32SH^3}$. We have that $W \geq W_1 = \frac{128 S^4 H^6 \log(2S / \delta')}{\epsilon^2}$, and the empirical transition distribution for a  window of $BW$ steps at a state $(x, h)$ has total variation distance with $q_{\hat{V}}(\cdot ; x, h)$ at most $\frac{\epsilon}{32SH^3}$ with probability at least $1-\delta'$.
Recall that 
$\delta' \leq \frac{\epsilon \delta}{192 SH^4 \parens{(S+1)^H + 1} \cdot \max\parens{S , 4 H^7/{\epsilon}}} = \frac{\delta BW}{3 LH ((S+1)^H + 1)}$;  
there are $LH$ total {\it steps} in each epoch, which fall into at most $\frac{LH}{BW}$ completed windows of length $BW$, and so the above holds for all windows in an epoch with probability at least $1 - \frac{\delta}{3 ((S+1)^H + 1) }$ by a union bound.
The bound then holds for every pair and epoch with probability at least $1 - \delta / 3$.


Using bounds on the empirical outgoing visitation distributions for each pair which is visited sufficiently often, we can obtain a bound on the
total variation distance between $q_{\hat{V}}$ and the empirical visitation distribution over the epoch at each step,
by \Cref{lemma:tvd-composition}.
All but at most ${2SWB}$ of the steps fall into separate but contiguous windows of length $WB$, as there can be at most two ``incomplete'' windows (at the start and end) for each state where we cannot apply the above analysis. 
Observe that accounting for these unfinished windows increases the total variation distance between $q_{\hat{V}}$ and the empirical visitation distribution by at most $\frac{\epsilon}{32SH^3}$ if $\frac{ 2 SW B}{L} \leq \frac{\epsilon}{32SH^3}$, as this bounds the fraction of trajectories in which our original bound does not apply.
This is the case when $L \geq \frac{64 S^2 H^3 W_1 B}{\epsilon}$.
It follows that the total variation distance between $q_{\hat{V}}$ and the empirical visitation distribution increases by at most $\frac{\epsilon }{16SH^3}$ for each step in $[H]$.
If the total variation distance with $q_{\hat{V}}(\cdot ,h)$ is at most $\frac{\epsilon h}{16SH^3}$ at each step, then any state
with $q_{\hat{V}}(x, h) \cdot L$ expected visits gets at least $\parens{q_{\hat{V}}(x, h) -  \frac{\epsilon h}{16SH^3}}\cdot L$ visits.

Each state with $q_{\hat{V}}(x, h) \geq \frac{\eps}{8H^2}$ 
is therefore visited at least $\frac{\epsilon }{16SH^2}\cdot L$ times when the above events hold.
States are locked after $\frac{16H^2 W B}{\epsilon}$ visits;
as such, if $L \geq \max\parens{\frac{64 S^2 H^3 W_1 B}{\epsilon} , \frac{256 SH^4 W B}{\epsilon^2}}$ all states with $q_{\hat{V}}\parens{x, h} \geq \frac{\eps}{8SH^2} $ are visited enough to be locked in the epoch.

\end{proof}

We now have that
if a state has mass at least $\frac{\eps}{8SH^2}$ under $q_{\hat{V}}$, it will be visited frequently enough to be locked in the epoch corresponding to value estimates $\hat{V}$, with high probability. Contrapositively, when this holds it implies that if a state is unlocked (but not reset) after the termination of an epoch, it must have had small mass under $q_{\hat{V}}$ for that epoch. 


Let $U_h = \{x \mid (x, h) \text{ is unlocked at termination}\}$.
We can then use a similar inductive argument (\Cref{lemma:unlocked-mass-fhsg}) to show that unlocked states have small mass under $q_D$ at termination. An important step here is in bounding the total variation distance with $p_{\hat{V}}(\cdot; x, h)$, which we do in \Cref{lemma:batch-tvd-fhsg}.

\begin{lemma}\label{lemma:batch-tvd-fhsg}
Let $D_{W,\hat{V}}(x,h)$ be a set of action profiles at a pair $(x, h)$ generated by $W$ completed runs of $\B$ for all players.With probability at least $1- \delta'$, the total variation distance between $p_{\hat{V}}(\cdot; x,h)$ and (transition distribution given $a \sim D_{W, \hat{V}}(x,h)$) is at most $\frac{\epsilon}{32 SH^3}$.
\end{lemma}
\begin{proof}

Each run of $\B$ generates a sequence of action profiles; for each action profile, there's some fixed probability that a state $x'$ will be visited next. Whether or not this state is actually visited is a $[0,1]$ random variable with some expected value.
The number of realized visits to $x'$ versus the expected number of visits given the action profile can be expressed as a martingale, and as such the expectation over profile generation and transition function resampling is equal to the expected number of visits.
Note that this number of visits to $x'$ in a run of $\B$ is itself a random variable with mean $B \cdot p_{\hat{V}}(x' ; x, h)$.
and so $\E[p_{D, W}(x' ; x, h)] = p_{\hat{V}}(x' ; x, h)$.
We can then apply the same concentration analysis as in \Cref{lemma:visit-freq-fhsg} to give us that the total variation distance between $p_{D, W}(\cdot ; x, h)$ and $p_{\hat{V}}(\cdot ; x, h)$ is at most $\frac{\epsilon}{32SH^3}$ with probability $1 - \delta'$.

\end{proof}

We now have that \Cref{lemma:visit-freq-fhsg} and \Cref{lemma:batch-tvd-fhsg} hold for every window across all epochs with probability at least $1 - 2\delta / 3 $ by a union bound. The union-bound analysis for when \Cref{lemma:batch-tvd-fhsg} holds for all epochs and pairs is equivalent to that for \Cref{lemma:visit-freq-fhsg}.

\begin{lemma}\label{lemma:unlocked-mass-fhsg}
    When the algorithm terminates, 
    with probability at least $1 - \frac{2\delta}{3}$,
    for each step $h$
    $\sum_{x \in U_h } q_{D}(x,h) \leq \frac{\epsilon h}{4H}$.
\end{lemma}
\begin{proof}

When
all events for events for \Cref{lemma:visit-freq-fhsg} occur for all epochs (at most $(S+1)^H +1$), 
any state which is unlocked and not reset after the end of an epoch must have $q_{\hat{V}}(x,h) \leq \frac{\epsilon}{8SH^2}$ for the corresponding $q_{\hat{V}}$. For the final epoch and its set of value estimates for all agents $\hat{V}$, this means that any unlocked state $(x,h)$ has $q_{\hat{V}}(x,h) \leq \frac{\eps}{8SH^2}$ at termination, and so $\sum_{x \in U_h } q_{\hat{V}}(x,h) \leq \frac{\epsilon}{8H^2}$ for each $h$.

Immediately we have that the lemma holds for all pairs $(x, 1)$, as their probabilities are defined identically under $q_{D}$ and $q_{\hat{V}}$.

From \Cref{lemma:batch-tvd-fhsg}, we can see that for every locked state $(x,h)$, we have that
$d_{TV}(p_{\hat{V}}(\cdot ; x', h), p_{D}(\cdot ; x', h)) \leq \frac{\eps}{8H^2}$.
Because we complete $\frac{16H^2 W}{\epsilon}$ runs of $\B$ before locking any state, 
the total variation distance between $p_{D}(\cdot ; x, h)$ and $p_{\hat{V}}(\cdot ; x, h)$ is at most $\frac{\epsilon}{32SH^3} + \frac{\epsilon}{16H^2} \leq \frac{\epsilon}{8H^2}$, assuming worst-case total variation distance for the final sequence of up to $BW$ trajectories for which our bound does not apply.
Further, each unlocked state has mass at most $\frac{\eps}{8SH^2}$ under $q_{\hat{V}}$. We can bound the total variation distance between $q_D$ and $q_{\hat{V}}$ at each step in terms of earlier steps as well as the distance from $p_{\hat{V}}$ for each pair's outgoing transition distribution using \Cref{lemma:tvd-composition}.

Expanding out, we can explicitly bound the total variation distance at each step, using the fact that the distributions are identical for $h=1$. With $d_{q,V,D}^h = d_{TV}(q_{\hat{V}}(\cdot , h), q_{D}(\cdot , h))$:
\begin{align*}
    d_{q,\hat{V},D}^h  \leq&\; \sum_{j=1}^{h-1} \sum_{x' \in X}  q_{\hat{V}}(x' , j ) \cdot d_{TV}(p_{\hat{V}}(\cdot ; x', j), p_{D}(\cdot ; x', j)) \\ 
    \leq&\; \sum_j^{h-1} \parens{ \sum_{x' \in U_j} q_{\hat{V}}(x' , j ) +  \sum_{x' \in \X \setminus U_j}   q_{\hat{V}}(x' , j ) \cdot \frac{\epsilon}{8H^2} } \\ 
    \leq&\;  \sum_j^{h-1} \parens{\frac{\epsilon}{8H^2}+\frac{\epsilon}{8H^2}  } \\
    =&\; \frac{\epsilon(h-1)}{4H^2}.
\end{align*}

Applying this to our bound on the mass of unlocked states under $q_{\hat{V}}$ completes the proof of the lemma:
\begin{align*}
    \sum_{x \in U_h } q_{D}(x,h) =&\; \sum_{x \in U_h } q_{\hat{V}}(x,h) + \sum_{x \in U_h }  q_{D}(x,h) - q_{\hat{V}}(x,h) \\
    \leq&\; \frac{\epsilon}{8H^2} + d_{TV}(q_{\hat{V}}(\cdot , h), q_{D}(\cdot , h)) \\ \leq&\; \frac{\epsilon h}{4H^2}.
\end{align*}

\end{proof}

We conclude by bounding the regret when agents play according to $D$. First we analyze the regret each agent playing according to $D$ under the assumption that all agents receive the maximal reward $H-h+1$ for the remainder of the trajectory upon reaching a state in $U_h$. We show that this is small, and that it does not increase by much upon correcting for the unlocked states.

It will be convenient for us to consider regret with respect to function classes $\Fsw_i^h : \A_i \times \X \times [h, \ldots, H] \rightarrow \A_i$, which we deem $\Fsw^h$-regret. This is in $[0,H-h+1]$ denoting the maximum possible downstream per-trajectory improvement by a swap function which only changes behavior in steps $h$ and onwards.
Because we complete at least $W$ runs of $\B$ before locking each state,
we can apply the guarantees of Corollary 3.1, where $B = B\parens{\frac{\eps}{16H}, N}$ and $\eta = \frac{\eps}{16 H^2}$ at each pair, which holds simultaneously for all pairs and players with probability $1-\delta/3$ by a union bound, giving us a total failure probability of at most $\delta$.
For pairs at step $H$, which are equivalent to games with stochastic rewards, this gives us that
\begin{itemize}
    \item the ``local'' $\Fsw^H$-regret for a pair $(x, H)$ is at most $\frac{\eps}{4H} + \frac{\epsilon}{32 H^2}$, and
    \item the estimated value is within $\frac{\eps}{16 H^2}$ of the true expected average value of running the bandit algorithm at that pair.
\end{itemize}
For steps $h < H$ these hold as well, but scaled by a factor of $H-h+1$, under the assumption that estimates of pair values reflect the true expected value of being at that pair. The corresponding distribution over reward tensors for the implicitly represented game with stochastic rewards can be obtained by taking the product distribution over transition functions and reward tensors, then converting each transition-reward pair to a tensor by adding each players' value estimates for visited states at the next step to their utility (recall that rewards and transitions are independent). We will later account for this estimation error.

For every pair, there can be up to $W$ runs of $\B$ at termination for which this bound doesn't hold, but otherwise we can average the contiguous sequences of $W$ and apply the same bounds for value and regret.
Because we complete at least $\frac{16 WH^2}{\epsilon}$ runs of $\B$ before locking a state, even assuming maximal average regret for this subsequence, the total average regret increases by at most $\frac{\epsilon}{16 H^2}$. The same error bound applies to value estimates.

We can then show that computed value estimates will not be far from the true expected downstream utility of that state when all agents play the correlated equilibrium. If we can bound the estimation error for downstream pairs at step $h+1$, the estimation error at step $h$ is bounded by the sum of the ``local'' and downstream error. We let $d(j)$ denote this bound for locked pairs $(x, H-j+1)$: 
\begin{align*}
    d(1) \leq&\; \frac{\epsilon}{8 H^2}, \\
    d(j) \leq&\; \frac{\epsilon j }{8 H^2} + d(j-1) \\ 
    =&\; \sum_{i=1}^j \frac{\epsilon i }{8 H^2} \\
    =&\; \frac{j(j+1)}{2}\cdot \frac{\epsilon }{8 H^2} 
\end{align*}
We can also bound the regret of the distribution in a similar manner. Suppose each value estimate downstream from some pair $(x,h)$ was exactly accurate, and each such downstream subgame had no regret; then the local regret (from the copy of $\B$) constitutes the entire subgame regret. 
Regret increases by at most twice the downstream error bound (recall we are assuming for now that this bound applies to locked and unlocked states), as this 
bounds the amount that any pair of swap functions $f, f' \in  \Fsw^h$ (including $I$) can deviate in the difference of their utilities when considering average reward from playing the game according to the specified action distributions. 
Finally, we add the downstream regret.  
As such, the 
following expression bounds the total regret at a pair: 
\begin{align*}
     (x,h) \text{ regret} \leq&\; \text{local regret} + 2 \times \text{downstream error} \\
     &\;+ \text{downstream regret}
\end{align*}
Here, all terms are defined with respect to the resulting distribution of profiles and the true distribution over rewards and transitions the game. 
We let $\hat{r}(j)$ denote the total regret (under $q_D$, assuming maximal reward from unlocked states) at a pair at step $j = H-j+1$ and let $\ell(j)$ denote the local regret. For each, we have that
\begin{align*}
    \ell(j) =&\; j \parens{ \frac{\epsilon}{4H} + 2 \cdot \frac{\epsilon }{16H^2}  }
\end{align*}
and so total regret is bounded by 
\begin{align*}
    \hat{r}(j) \leq&\; \ell(j) + 2d(j-1) + \hat{r} (j-1) \\
    =&\; \ell(j) + \sum_{i=1}^{j-1} \ell(i) + 2d(i) \\ 
    \leq&\; j \parens{ \frac{\epsilon}{4H} +  \frac{\epsilon}{8H^2} } + \sum_{i=1}^{j-1} i\parens{ \frac{\epsilon}{4 H} + \frac{\epsilon}{8H^2} + \frac{(i+1)\eps}{8H^2}} \\ 
    \leq&\; (j + j^2/2) \parens{ \frac{\eps }{4 H} + \frac{\eps }{8H^2} } + \sum_{i=1}^{j-1} \frac{(i^2+i)\eps}{8H^2}  \\ 
    \leq&\; (j + j^2/2) \parens{ \frac{\eps }{4H} + \frac{\epsilon}{8H^2} } + \frac{j^3 \epsilon} {24 H^2}. 
\end{align*}
For $j=H$, corresponding to the regret bound for each state at step 1, we have that 
\begin{align*}
    \hat{r}(H) \leq&\; (H + H^2/2) \parens{ \frac{\eps }{4H} +  \frac{\epsilon}{8H^2} } + \frac{\epsilon H}{24}  \\
    \leq&\; \frac{\epsilon}{4} +  \frac{\epsilon}{16} +  \frac{\epsilon H}{8} + \frac{\epsilon}{8 H} + \frac{\epsilon H}{24}. \\
\end{align*}
All of the (maximal) value estimates for unlocked states are overestimates; because no swap function can improve average expected utility by more than the above bound before correcting for unlocked states, we can use the frequency of unlocked states to bound the true regret. 
If all unlocked states at step $h$ have  $\sum_{x \in U_h} q_D(x, h) \leq \frac{h\epsilon}{4H^2}$, their contribution to the average regret of $D$ is bounded by
\begin{align*}
    \sum_{h=1}^H (H-h+1)  \frac{h\eps}{4H^2}  =&\; \frac{H(H+1)(H+2)}{6} \cdot \frac{\epsilon}{4H^2}\\
    \leq&\; \frac{\epsilon H}{4}.
\end{align*}
Adding in the maximal contributions from unlocked states, we have that
\begin{align*}
    {r}(H) \leq&\; \frac{\epsilon}{4} +  \frac{\epsilon}{16} +  \frac{\epsilon H}{8} + \frac{\epsilon}{8 H} + \frac{\epsilon H}{24} + \frac{\epsilon H}{4} \\
    \leq&\; 0.855 \epsilon H
\end{align*}
for all $\epsilon \leq 1$ and $H \geq 1$. As this bound holds simultaneously at each pair $(x, 1)$ for all players, and captures the expected regret over an entire trajectory when players play according to $D$, the average $\Fsw$-regret
{\it per step of the game} is less than $\epsilon$. Thus, the policy distribution constitutes an $\epsilon$-EFCE for the game.

\end{proof}

\subsection{Analysis for Fast PLL}
\label{sec:proofs-fast-fhsg}
We restate the description of FastPLL with $L$ given precisely.
\paragraph{Algorithm 4: Fast PLL.}
Let $B = B\parens{\frac{ \epsilon }{8H}, N}$, and the epoch length (in trajectories) be given by 
\begin{align*}
    L \geq&\; \frac{ \sqrt{{\log(2SH / \delta) } + \frac{1024 BH^4 \gamma \log\parens{\frac{10SHM}{\delta}}  }{\epsilon^2}} }{4 \gamma^2} \\
    &\quad+ \frac{\log(2SH / \delta) + \frac{512 B H^4 \gamma \log\parens{\frac{10SHM}{\delta}}  }{\epsilon^2} }{4 \gamma^2}.
\end{align*}
Run $H$ epochs, one corresponding to each step (beginning with step $H$) as follows:

\begin{itemize}
    \itemsep0em
    \item \emph{Epoch for Step $h$:} Use a copy of $\B$ to select actions at each pair $(x, h)$, augmenting rewards with computed values for pairs $(x', h+1)$ transitioned to for the next step (if $h < H$). At the end of the epoch, let $\hat{V}_i^{\B}(x, h)$ be the average reward received from all completed runs of $\B$.
    \item \emph{Upstream ($h' < h$):} Select actions uniformly at random for each pair. 
    \item \emph{Downstream ($h' > h$):} Use $\B$ at each signal as in the epoch for step $h'$, 
        augmenting rewards with value estimates for pairs transitioned to.
      Restart $\B$ after every $B$ rounds in which it is used, which can include rounds from a prior epoch. 
\end{itemize}

\paragraph{Restatement of Theorem 6.}\emph{    After Algorithm 3 terminates, for each pair $(x, h)$, consider the uniform distribution over action profiles $D(x,h)$ played since epoch $H-h+1$ began. 
    Let $D$ be the distribution over policy profiles where the action profile for each pair $(x,h)$ is sampled independently from $D(x, h)$. 
With probability at least $1 - \delta$, $D$ is an $\epsilon$-correlated equilibrium for the game.}

\begin{lemma}
    \label{lemma:fast-visits}
    With probability at least $1 - \delta/2$, every state $x$ is visited at step $h$ at least $\frac{128BH^4 \log\parens{\frac{10 S H M}{\delta}} }{\epsilon^2}$ times in epoch $H-h+1$.
\end{lemma}
\begin{proof}
Fix some pair $(x, h)$.
Let $X = \sum_{i=1}^L X_i$ be a sum of indicator random variables denoting the number of times $(x,h)$ is visited in epoch $H-h+1$. 
By the fast-mixing assumption, $\EE{X} \geq \gamma L$.
For $L$ as specified, we have that
\begin{align*}
    \gamma L - \sqrt{\frac{L \log (2 S H) }{2}} \geq&\; \frac{128BH^4 \log\parens{\frac{10 S H M}{\delta}} }{\epsilon^2}
\end{align*}
by the quadratic formula.
By Hoeffding's inequality, with $Y$ being the event where $X \leq  \frac{128BH^4 \log\parens{\frac{10 S H M}{\delta}} }{\epsilon^2}$:
\begin{align*}
    \pr{Y}  =&\; \pr{ \E[X] - X \geq \E[X]  - \frac{128BH^4 \log\parens{\frac{10 S H M}{\delta}} }{\epsilon^2}} \\
    \leq&\; \pr{ \E[X] - X \geq \gamma L  - \frac{128BH^4 \log\parens{\frac{10 S H M}{\delta}} }{\epsilon^2}} \\
    \leq&\; \pr{ \E[X] - X \geq \sqrt{\frac{L \log(2SH / \delta) }{2} } } \\
    \leq&\; \exp\parens{-\log(2SH / \delta) }  \\
    \leq&\; \frac{\delta}{2SH},
\end{align*}
and the lemma follows from union-bounding over all pairs.
\end{proof}

\begin{proof}[Proof of Theorem 6]
    First we see that after epoch 1, the value estimates $\hat{V}_i(x, H)$ are within $\frac{}{}$
$B(\frac{\epsilon}{8H}, N)$
Let $\eta = \eps/8 H^2$. 
From \Cref{lemma:fast-visits}, each pair is visited at least $\frac{128 B H^4 \log\parens{\frac{10SHM}{\delta}} }{\epsilon^2}$ times in its corresponding epoch with probability at least $1 - \delta / 2$.
When this holds, we can apply the guarantees of Corollary 3.1, where $B = B\parens{\frac{\eps}{8H}, N}$ and $\eta = \frac{\eps}{8 H^2}$ at each pair, which holds simultaneously for all pairs and players with probability $1-\delta/2$ by a union bound, giving us a total failure probability of $\delta$.
For pairs at step $H$, which are equivalent to games with stochastic rewards, this gives us that
\begin{itemize}
    \item the ``local'' $\Fsw^H$-regret for a pair $(x, H)$ is at most $\frac{\eps}{2H} + \frac{\epsilon}{16 H^2}$, and
    \item the estimated value is within $\frac{\eps}{8 H^2}$ of the true expected average value of running the bandit algorithm at that pair.
\end{itemize}
Again for steps $h < H$ these hold as well, scaled by a factor of $H-h+1$, under the assumption that estimates of pair values reflect the true expected value of being at that pair. We will account for this estimation error below.

Note that we can take these bounds to hold after all epochs terminate rather than simply the corresponding epoch. 
This is because neither the algorithm nor downstream values change for each step in future epochs once its value is computed. This ignores the sole possibly truncated run of $\B$ when the final epoch terminates. Assuming maximal average regret for this subsequence, the total average regret increases by at most 
$\frac{\epsilon^2}{128 H^4 \log\parens{ \frac{10SHM} {\delta} } } \leq \frac{\epsilon^2}{128 H^4}$ given the number of resets of $\B$ per epoch. The same error bound applies to value estimates.

We can then show that computed value estimates will not be far from the true expected downstream utility of that state when all agents play the correlated equilibrium. If we can bound the estimation error for downstream pairs at step $h+1$, the estimation error at step $h$ is bounded by the sum of the ``local'' and downstream error. We let $d(j)$ denote this bound for pairs $(x, H-j+1)$: 
\begin{align*}
    d(1) \leq&\; \frac{\epsilon}{8 H^2} + \frac{\epsilon^2}{128 H^4} \\
    d(j) \leq&\; \frac{\epsilon j }{8 H^2} + \frac{\epsilon^2 j}{128 H^4}+ d(j-1) \\ 
    =&\; \sum_{i=1}^j \frac{\epsilon i }{8 H^2} + \frac{\epsilon^2 i}{128 H^4} \\
    =&\; \frac{j(j+1)}{2}\cdot \parens{ \frac{\epsilon }{8 H^2} + \frac{\epsilon^2}{128 H^4}}  \\
    \leq&\; j^2  \parens{\frac{\epsilon }{8 H^2} + \frac{\epsilon^2 }{128 H^4} } \\ 
\end{align*}
We can also bound the regret of the distribution in a similar manner. As in
the proof of Theorem 5,
the total regret at a pair can be bounded as: 
\begin{align*}
    (x,h) \text{ regret} \leq&\; \text{local regret} + 2 \times \text{downstream error} \\
    &\;+ \text{downstream regret}
\end{align*}

Here, all terms are defined with respect to the resulting distribution of profiles and the true distribution over rewards and transitions for the game. 
We let $r(j)$ denote the total regret at a pair at step $j = H-j+1$ and let $\ell(j)$ denote the local regret. For each, we have that
\begin{align*}
    \ell(j) \leq&\; j \parens{ \frac{\epsilon}{2H} + \frac{\epsilon}{16H^2} + \frac{\epsilon^2 }{128 H^4}  }
\end{align*}
and so total regret is bounded by 
\begin{align*}
    r(j) \leq&\; \ell(j) + 2d(j-1) + r(j-1) \\
    =&\; \ell(j) + \sum_{i=1}^{j-1} \ell(i) + 2d(i) \\ 
    \leq&\; j \parens{ \frac{\epsilon}{2H} +  \frac{\epsilon}{16H^2} +  \frac{\epsilon^2 }{128 H^4}  } \\
    &\;+ \sum_{i=1}^{j-1} i\parens{ \frac{\epsilon}{2 H} + \frac{\epsilon}{16H^2} + \frac{\eps i}{4 H^2} + \frac{ (2i + 1) \epsilon^2 }{128 H^4}  } \\ 
    \leq&\; (j + j^2/2) \parens{ \frac{\eps }{2 H} + \frac{\eps }{16H^2} + \frac{\epsilon^2 }{128 H^4}  } \\
    &\;+ \sum_{i=1}^{j-1} \parens{ \frac{\eps i^2}{4 H^2} + \frac{  \epsilon^2 i^2 }{64 H^4}  } \\ 
    \leq&\; (j + j^2/2) \parens{ \frac{\eps }{2H} + \frac{\epsilon}{16H^2} +  \frac{\epsilon^2 }{128 H^4}  } \\
    &\;+ \frac{j^3}{3} \parens{ \frac{\eps}{4 H^2} + \frac{  \epsilon^2 }{64 H^4}  } \\ 
\end{align*}
For $j=H$, corresponding to the regret bound for each state at step 1, we have that 
\begin{align*}
    r(H) \leq&\; (H + H^2/2) \parens{ \frac{\eps }{2H} +  \frac{\epsilon}{16H^2} + \frac{\epsilon^2 }{128 H^4}  } \\
    &\;+ \frac{H^3}{3} \parens{ \frac{\eps}{4 H^2} + \frac{  \epsilon^2 }{64 H^4}  } \\
    \leq&\; \frac{\epsilon}{2} +  \frac{\eps}{32} +  \frac{\epsilon H}{4} + \frac{\epsilon H}{12} + \frac{\eps}{16H} \\
    &\;+  \frac{\epsilon^2}{192 H} + \frac{\epsilon^2}{256 H^2} + \frac{\epsilon^2}{128 H^3}  \\
    \leq&\; 0.945 \epsilon H.
\end{align*}
As this bound holds simultaneously at each pair $(x, 1)$ for all players, and captures the expected regret over an entire trajectory, the average $\Fsw^1$-regret (equivalent to $\Fsw$-regret)
per step of the game is less than $\epsilon$. Thus, the policy distribution constitutes an $\epsilon$-correlated equilibrium for the game.
\end{proof}
\subsection{Analysis for Single-Controller Stochastic Games}
\label{sec:proofs-sc-fhsg}

\paragraph{Restatement of Theorem 7.}
\emph{With probability at least $1 - \delta$, the uniform distribution over the sequence of policy profiles played by Algorithm 4 is an $\epsilon$-NFCCE for the game.}

\begin{proof}[Proof of Theorem 7]
    The theorem follows directly from \Cref{lemma:sc-controller-regret} and \Cref{lemma:sc-follower-regret}.
\end{proof}

\begin{lemma}\label{lemma:sc-controller-regret}
    After $T\geq \frac{8 B_L(\epsilon/8) \log(M/\delta) }{\epsilon^2}$ trajectories, the controller has average regret $\epsilon H$ per trajectory with probability at least $1-\delta/M$.
\end{lemma}
\begin{proof}[Proof of \Cref{lemma:sc-controller-regret}]
Consider the sampled reward tensors (for every pair) in each trajectory.
 When all followers select policies in this trajectory, 
 the current task for the controller is equivalent to an MDP (consider the fixed distribution of transitions for each action, identical across trajectories, defined by $p$). 
 The task for the controller is equivalent to that of optimizing over MDPs with unknown but fixed transitions and adversarial losses; an expected per-trajectory regret bound of $\epsilon H/8$ for the {\it policy class} follows from Theorem 7.2 of \cite{Rosenberg2019OnlineSS} with the appropriate polynomial runtime (obtainable from inverting their regret bound), holding with respect to the set of tensors sampled in that round.
 Their state count corresponds to $SH$ in our setting, as they assume a ``loop-free'' episodic MDP, which can be created from any MDP with an increase by a factor of at most $H$ for the state space.

As we saw in the analysis of Theorem 3, we can again view the performance difference for each policy on the {\it realized} and {\it expected} sequence of sets of reward tensors as a martingale --- given opponent policies, the reward received in the trajectory by any policy is a random variable.
If $T \geq \frac{128\parens{SH \log(N) + \log(16/\epsilon)}}{\epsilon^2}$, then by Azuma-Hoeffding the probability that a policy's per-step reward deviates more than $\frac{\epsilon}{8}$ from expectation is at most $\frac{\epsilon}{8N^{SH}}$.
As in the analysis of Theorem 3, by chaining deviation bounds and union-bounding over all $N^{SH}$ policies, it then follows that the {\it expected} policy regret for the sequence of policy profiles, given the distribution of rewards and transitions at each state, is at most $\epsilon/2$.
Given the runtime of Shifted Bandit U-CO-REPS, $T$ is sufficiently large for this to hold extending the runtime as we did in Theorem 3.
As such, for the policy sequence over $\B_L(\epsilon/8)$ the expected average per-step regret for the controller when sampling reward tensors and transition functions independently at each state is at most $\epsilon/2$.

Again, this is boosted to $\epsilon$ average regret with probability $1 - \frac{\delta}{M}$ after repeating for $\frac{8 \log(M/\delta)}{\epsilon^2}$ such sequences, at which point 
the average regret is at most $\frac{3\epsilon}{4}$ with probability at least $1 - \frac{\delta}{M}$ by Hoeffding's inequality. If $T$ is some arbitrary fixed (but sufficiently large) number of trajectories, there may be at most one run of length $B_L(\epsilon/8)$ which is {\it incomplete}, in that we cannot apply the above analysis; however, even assuming maximum regret across this sequence, the total average regret increases by at most $\frac{\epsilon^2}{8 \log(M/\delta) } < \epsilon/4$, completing the proof.
\end{proof}

\begin{lemma}\label{lemma:sc-follower-regret}
After $T\geq \frac{8 B_F(\frac{\epsilon}{8}) \log(M/\delta)}{\epsilon^2 }$ trajectories, every follower has average swap regret across all pairs of at most $\epsilon H$ per trajectory with probability at least $1 - \frac{\delta (M-1)}{M}$.
\end{lemma}
\begin{proof}[Proof of \Cref{lemma:sc-follower-regret}]

Followers run copies of Bayesian game algorithm $\B_S$ in parallel at each step, and the analysis largely follows from that in \Cref{sec:bayes}. The key ideas are to observe that regret can be decomposed stepwise (any deviations cannot affect transitions), and that we did not explicitly need the distribution over signals to be static in a Bayesian game, so long as our notion of regret tracks this shifting distribution.
The analysis of $\B_S$ in \Cref{thm:bayesian-games-ce} carries through directly if we consider a sequence of distributions over states and we aim for small regret with respect to this sequence, as we can equivalently define martingales to track deviations from expectation for each swap function at each step. 
As such, after $B(\frac{\epsilon}{8S}, N)$ trajectories, the {\it local} expected per-step regret is at most $\frac{\epsilon}{2}$ at each step with respect to the distribution over states induced by opponents' policies at that step. As swap regret bounds traditional regret and followers' actions don't affect transitions, the expected average regret per trajectory is at most $\frac{\epsilon H}{2}$, holding with respect to the randomness in the game.
Concentration analysis and handling truncation of a final sequence is equivalent to that in \Cref{lemma:sc-controller-regret}, and we union-bound over the $M-1$ followers.
\end{proof}

\subsection{Analysis of Simultaneous No-Regret with Shared Randomness}
\label{subsec:proofs-sr}

\paragraph{Restatement of Theorem 8}\emph{With respect to $\Fsw$,
PLL-SR has regret $\tilde{O}(T^{\frac{6}{7}})$ and FastPLL-SR has regret $\tilde{O}(T^{\frac{4}{5}})$.}

\begin{proof}[Proof of Theorem 8]
Let $T_{PLL}$ denote the maximum runtime of PLL (in {\it steps}), calibrated for an $\epsilon_1$-EFCE. Our choice of $\epsilon_1 = \tilde{\Theta}\parens{\sqrt[7]{\frac{N^3 S^{O(H)}}{T}}}$ is calibrated such that $T_{PLL} + \epsilon_1 (T - T_{PLL}) = \tilO(T^{\frac{6}{7}})$. Each step after termination is equivalent to playing according to the equilibrium PLL generates, as we are sampling action profiles independently across timesteps using the shared randomness (we can use the same random string to select actions at non-visited states at that step for the purposes of defining a full policy sequence).
Assuming a maximum per-step regret of 1 during the runtime of PLL (we can consider arbitrary ``policies'' for that window at pairs not visited in those trajectories, as PLL only chooses an action for visited pairs) and applying Theorem 5 to bound the regret for the remainder gives us the result for PLL-SR. The analysis for FastPLL-SR is symmetric.
\end{proof}

\end{document}